\documentclass[sigconf]{acmart}
\makeatletter

\usepackage{booktabs}
\usepackage{amsmath}
\usepackage{amsthm}
\usepackage{algorithmic}
\usepackage{algorithm}
\newtheorem{theorem}{Theorem}
\usepackage{graphicx}
\usepackage{amsmath}
\usepackage{url}
\usepackage{hyperref}
\usepackage{adjustbox}
\usepackage{color}
\usepackage{multirow}
\usepackage{setspace}
\usepackage{enumitem}
\usepackage{lipsum}

\acmYear{2018}
\copyrightyear{2018}
\setcopyright{usgovmixed}
\acmConference[HPDC '18]{HPDC '18: International Symposium on High-Performance Parallel and Distributed Computing}{June 11--15, 2018}{Tempe, AZ, USA}
\acmBooktitle{HPDC '18: International Symposium on High-Performance Parallel and Distributed Computing, June 11--15, 2018, Tempe, AZ, USA}
\acmPrice{15.00}
\acmDOI{10.1145/3208040.3208050}
\acmISBN{978-1-4503-5785-2/18/06}

\begin{document}
\title{Improving Performance of Iterative Methods by Lossy Checkponting}

\author{Dingwen Tao}
\affiliation{
\institution{University of California, Riverside}
\city{Riverside}
\state{CA}
\country{USA}
}
\email{dtao001@cs.ucr.edu}

\author{Sheng Di}
\affiliation{
\institution{Argonne National Laboratory}
\city{Lemont}
\state{IL}
\country{USA}
}
\email{sdi1@anl.gov}

\author{Xin Liang}
\affiliation{
\institution{University of California, Riverside}
\city{Riverside}
\state{CA}
\country{USA}
}
\email{xlian007@ucr.edu}

\author{Zizhong Chen}
\affiliation{
\institution{University of California, Riverside}
\city{Riverside}
\state{CA}
\country{USA}
}
\email{chen@cs.ucr.edu}

\author{Franck Cappello}
\affiliation{
\institution{Argonne National Laboratory}
\city{Lemont}
\state{IL}
\country{USA}
}
\email{cappello@mcs.anl.gov}

\thanks{Corresponding author:
Sheng Di, Mathematics and Computer Science Division, Argonne
National Laboratory, 9700 Cass Avenue, Lemont, IL 60439, USA}

%
%


\begin{abstract}
Iterative methods are commonly used approaches to solve large, sparse linear systems, which are fundamental operations for many modern scientific simulations.  When the large-scale iterative methods are running with a large number of ranks in parallel, they have to checkpoint the dynamic variables periodically in case of unavoidable fail-stop errors, requiring fast I/O systems and large storage space. 
To this end, significantly reducing the checkpointing overhead is critical to improving the overall performance of iterative methods.  
Our contribution is fourfold.
(1) We propose a novel lossy checkpointing scheme that can significantly improve the checkpointing performance of iterative methods by leveraging lossy compressors.
(2) We formulate a lossy checkpointing performance model  and derive theoretically an upper bound for the extra number of iterations caused by the distortion of data in lossy checkpoints, in order to guarantee the performance improvement under the lossy checkpointing scheme. (3) We analyze the impact of lossy checkpointing (i.e., extra number of iterations caused by lossy checkpointing files) for multiple types of iterative methods. (4) We evaluate the lossy checkpointing scheme with optimal checkpointing intervals on a high-performance computing environment with 2,048 cores, using a well-known scientific computation package PETSc and a state-of-the-art checkpoint/restart toolkit. Experiments show that our optimized lossy checkpointing scheme can significantly reduce the fault tolerance overhead for iterative methods by $23\%{\sim}70\%$ compared with traditional checkpointing and $20\%{\sim}58\%$ compared with lossless-compressed checkpointing, in the presence of system failures.
\end{abstract}
\maketitle

\section{Introduction}
\label{sec:intro}
Scientific simulations involving partial differential equations (PDEs) require solving sparse linear system within each timestep. 
At large scale, sparse linear systems are solved by using iterative methods, such as the conjugate gradient (CG) method. Thus, iterative methods are one of the most crucial components determining the scalability and efficiency of HPC applications. For example, Becciani et al. \cite{becciani2014solving} presented a study of solving a 5-parameter astrometric catalogue at the micro-arcsecond level for about 1 billion stars of our Galaxy under a cornerstone mission (called Gaia) launched by European Space Agency. Their experimental results show that solving the resulting sparse linear system of $7.2 \times 10^{10}$ equations for the last period of the Gaia mission can take 1,000 to 4,000 iterations for convergence, totaling up to $1.96 \times 10^5$ seconds (i.e., more than 54 hours) on 2,048 BlueGeneQ nodes.
When running on high-performance computing (HPC) environments using potentially tens of thousands of nodes and millions of cores for hours or days towards exascale computing \cite{parkere2013role}, fail-stop errors are inevitable.
Accordingly, how to effectively protect the iterative methods against such failures is an important research issue, determining the overall performance of iterative methods in HPC environments. 

Many algorithm-based fault tolerance approaches have been proposed to tolerate silent data corruptions  with iterative methods, and they work efficiently because of little storage overhead. Tolerating fail-stop errors, however, is much more challenging because it requires checkpointing or saving multiple large vector data sets at runtime, leading to large checkpointing overhead.

For many PDE-based scientific simulations, the sparse linear system includes most of the variables that are involved in the application, so checkpointing for  iterative methods determines overall checkpointing performance \cite{heath2002scientific}. For example, SIMPLE (Semi-Implicit Method for Pressure-Linked Equations) \cite{patankar1980numerical} algorithm is a widely used numerical method to solve the Navier-Stokes equations \cite{chorin1968numerical} for Computational Fluid Dynamics (CFD) problems. For 3D CFD problems, there are totally nine fluid-flow scalar variables, five of which need to be checkpointed during iterative methods. As a result, significantly improving the checkpointing performance of the iterative methods that it uses can significantly improve the application performance, since most of application state used by iterative methods. We refer readers to \cite{mora2001numerical} for more details of 3D CFD problems and SIMPLE-like algorithms in parallel.

In this work, we propose an efficient execution scheme, specifically a lossy checkpointing scheme, in order to improve the overall performance for iterative methods running in the presence of failures. Unlike the traditional checkpointing approach, our lossy checkpointing scheme integrates a lossy compression technique into the checkpoint/restart model. That is, the checkpointing data is compressed by a lossy compressor before being moved to the parallel file system (PFS), which is an approach that can significantly reduce the run-time checkpointing overhead.
Upon a failure, the latest checkpointing file is loaded and goes through a decompression step to reconstruct the checkpointing data for the recovery of the iterative execution. 

Checkpoint/restart research has been conducted for decades in order to optimize the performance of various large-scale scientific executions, but  lossy-compressed checkpointing is rarely studied. Lossy compressed checkpointing raises two challenging issues.
(1) What is the impact of lossy checkpointing data on the execution performance? Specifically, can the iterative methods still converge, or how many extra iterations will be introduced after restarting from a lossy checkpoint? (2) Is adopting lossy compression in the checkpointing model a worthwhile method for improving the overall performance? Specifically, how much performance gain can be achieved based on the checkpoints with reduced size? 

To address such two key issues, we make following contributions. 
\begin{itemize}
\item We propose a novel lossy checkpointing scheme that significantly improves the performance for iterative methods. In particular, we exploit a lossy checkpointing scheme under which both the lossy compression and checkpointing can be performed efficiently for the iterative methods. 
\item We design a performance model that can formulate the overall performance of the execution with lossy checkpointing in the presence of failures. In particular, we derive an upper bound for the extra number of iterations caused by the lossy checkpoints against the reduced checkpointing overheads, which is a sufficient condition to determine whether the lossy checkpointing can get a performance gain for an iterative method in numerical linear algebra.
\item We explore the impact of the lossy checkpointing on the extra number of iterations for multiple iterative methods, including stationary iterative methods, GMRES, and CG.
\item We evaluate our lossy checkpointing scheme with optimized checkpointing intervals based on multiple iterative methods provided by PETSc, using both lossless and lossy compressors, on a parallel environment with up to 2,048 cores. Experiments show that our solution reduces the fault tolerance overhead by $23\%{\sim}70\%$ compared with traditional checkpointing and $20\%{\sim}58\%$ compared with lossless checkpointing.
\end{itemize}

The rest of the paper is organized as follows. In Section \ref{sec:related}, we discuss  related work. In Section \ref{sec:checkpoint}, we describe the traditional checkpointing method without lossy compressors. In Section \ref{sec:lossy}, we propose our lossy checkpointing scheme with state-of-the-art lossy compression techniques included, and we provide an in-depth analysis of checkpoint/restart overhead and the impact of the lossy checkpointing on convergence. In Section \ref{sec:evaluation}, we present our experimental evaluation results. 
In Section \ref{sec:conclusion}, we conclude with a brief discussion of future work. 

\section{Related Work}
\label{sec:related}

Recently, a study of the Blue Waters system  \cite{bluewater} showed that an event that required remedial repair action occurred every 4.2 hours on average and that systemwide events occurred approximately every 160 hours. To avoid remedial actions such as redoing computations, researchers have designed many fault tolerance techniques for HPC applications \cite{liang2017correcting,wu2017silent, wu2016towards, chen2016online, di2016adaptive, wu2014ft, li2015applicationspecific}. 

\paragraph{Checkpoint/Restart Techniques} One of the most widely used techniques is the checkpoint/restart model, and the corresponding optimization strategies have been studied for many years.
Plank et al. \cite{diskless} proposed a diskless checkpointing approach that reduces the checkpoint overhead by storing the checkpoints locally in processor memories. 
However, diskless checkpointing can survive only partial failures: it is unable to deal with the failure of the whole system. 
A multilevel checkpoint/restart model \cite{scr,fti} was proposed to provide tolerance for different types of failures.
Fault Tolerance Interface (FTI) \cite{fti}, for example, supports four levels of checkpointing: local storage device, partner-copy, Reed-Solomon encoding, and PFS.
Di et al. \cite{di2014optimization-1, di2014optimization-2} proposed a multilevel checkpoint/restart model based on FTI to optimize the checkpoint intervals for different levels. In addition to the traditional checkpointing model, a few studies have demonstrated the feasibility of using compression techniques to improve the checkpointing performance.
Islam et al. \cite{islam} adopted data-aware aggregation and lossless data compression to improve the checkpoint/restart performance.
Sasaki et al. \cite{ssem} proposed a lossy compression technique based on wavelet transformation for checkpointing and explored its impact in a production climate application. 
Calhoun et al. \cite{calhoun} verified the feasibility of using lossy compression in checkpointing two specific PDE simulations experimentally. Their results show that the compression errors in the checkpointing files can be masked by the numerical errors in the discretization, leading to improved performance without degraded overall accuracy in the simulation.
To the best of our knowledge, our work is the first attempt to build a generic, theoretical performance model considering the impact of lossy compression techniques on the HPC checkpointing model and significantly improve the overall performance for multiple iterative methods, such as stationary iterative methods, GMRES, and CG. 

\paragraph{Fault Tolerance Techniques for Iterative Methods} Iterative methods are widely used for solving systems of equations or computing eigenvalues of large sparse matrices. Although some fault-tolerant iterative methods have been designed, most are from an algorithmic level, and the performance is highly dependent on the specific characteristics of algorithms.
For example, Tao et al. \cite{newsum} proposed an online algorithm-based fault tolerance (ABFT) approach to detect and recover soft errors for general iterative methods.
For some specific iterative algorithms, 
Chen \cite{chen2013online} developed an online ABFT approach for a subset of the Krylov methods by leveraging the orthogonality relationship of two vectors.
Bridges et al. \cite{bridges2012fault} and Elliot et al. \cite{ftgmres} targeted GMRES based on its special characteristics and proposed a fault-tolerant version via selective reliability. 
Similar to that work, Sao and Vuduc \cite{sao2013self} studied self-stabilizing corrections after error detection for CG algorithm. 
For fail-stop failures, Langou et al. \cite{langou} designed an algorithm-based recovery scheme for iterative methods, called lossy approach, that recovers an approximation of the lost data, but it is limited to the block Jacobi algorithm.
Chen \cite{chen2011algorithm} proposed an algorithm-based recovery method that utilizes inherent redundant information for accurately recovering the lost data, but it is limited to the memory failure situation.
Agullo et al. \cite{agullo2013towards} proposed a technique that can recover from process failures followed by restarting strategies in Krylov subspace solvers where lost entries of the iterate are interpolated to define a new initial guess before restarting the Krylov method.
Asynchronous iterations \cite{bahi2007parallel} proposed by Bahi et al. are linear solvers designed to
tolerate message delays when applying the matrix in parallel. 

\paragraph{Scientific Data Compression} Scientific data compression has been studied
for years. The data compressors can be split into two categories: lossless and lossy. Lossless compressors make sure that the reconstructed
data set after the decompression is exactly the same
as the original data set. Such a constraint may significantly
limit the compression ratio (up to 2 in general \cite{lossless2006}) on the compression of scientific
data. The reason is that scientific data are composed mainly of floating-point values and their tailing mantissa bits could be too random to compress effectively.  State-of-the-art lossy compressors include SZ \cite{sz16, sz17}, ZFP \cite{zfp}, ISABELA \cite{isabela}, FPZIP \cite{fpzip}, SSEM \cite{ssem}, and NUMARCK \cite{numarck}. Basically, they can be categorized
into two models: prediction based  and transform based. A prediction-based compressor predicts data values for each data point and encodes the difference between every predicted value and its corresponding real value based on a quantization method. Typical examples are SZ \cite{sz16, sz17}, ISABELA \cite{isabela}, and FPZIP \cite{fpzip}. 
The block-transform-based compressor transforms the original data to another space where most of the generated data
is very small (close-to-zero), such that the data can be stored with a certain loss in terms of user-required error bounds.
For instance, SSEM \cite{ssem} and ZFP \cite{zfp} adopt a discrete Wavelet transform and a customized orthogonal transform, respectively.
Lossy compression techniques, however, are  used mainly for saving storage space and reducing the I/O cost of dumping the analysis data. How to make use of the lossy compressors to improve the checkpointing performance with iterative methods is still an open question.
\section{Traditional Checkpointing Technique for Iterative Methods}
\label{sec:checkpoint}

Before presenting our lossy checkpointing scheme, we investigate the traditional checkpointing techniques for iterative methods.

According to a  study of recovery patterns for iterative methods by Langou et al. \cite{langou}, we need to classify the variables of the algorithms in order to form a fault-tolerant iterative method with the checkpoint/recovery model. All the variables can be categorized into three types: 
\begin{itemize}
\item Static variables: need to be stored once, for example, the system matrix $A$, the preconditioner matrix $M$, and the right-hand side vector $b$;
\item Dynamic variables: change along the iterations, for example, the approximate solution vector $x^{(i)}$;
\item Recomputed variables: are worth being recomputed after a failure rather than being checkpointed; for example, the residual vector $r$ can be recomputed by $r^{(i)} = b - Ax^{(i)}$). The term ``worth'' here means that recomputing some variables could be faster than obtaining them through a checkpoint.
\end{itemize}
 
Although the recomputed variables also need to be recovered during restarting after failures/errors, we still classify them as a separate category because they are recovered by a different strategy. How to recover a variable depends on the recovery overheads of the particular strategy. A scalar computed through global vector dot product, for example, is too expensive to compute, so it will be treated as a dynamic variable during the checkpointing.

After the classification is finished, we can form the fault-tolerant iterative methods with the checkpoint/recovery model as follows.
\begin{itemize}
\item Checkpoint
\begin{enumerate}
\item Checkpoint static variables only at the beginning before going into the execution with iterations,
\item Checkpoint dynamic variables every several iterations.
\end{enumerate}
\item Recovery
\begin{enumerate}
\item Recover a correct computational environment,
\item Recover static variables,
\item Recover dynamic variables,
\item Recover recomputed variables based on the reconstructed static and dynamic variables.
\end{enumerate}
\end{itemize}

\begin{algorithm}
\caption{Fault-tolerant preconditioned conjugate gradient (PCG) algorithm with traditional checkpointing.}
\renewcommand{\algorithmiccomment}[1]{/*#1*/}
\begin{flushleft}
\textbf{Input}: linear system matrix $A$, preconditioner $M$, and right-hand side vector $b$\\
\textbf{Output}: approximate solution $x$
\end{flushleft}

\begin{algorithmic}[1]
\STATE Compute $r^{(0)} = b - Ax^{(0)}$, $z^{(0)} = M^{-1}r^{(0)}$, $p^{(0)} = z^{(0)}$, $\rho_0 = {r^{(0)}}^Tz^{(0)}$ for some initial guess $x^{(0)}$
\FOR{$i = 0, 1, \cdots$}
\IF {(($i > 0$) and ($i\%ckpt\_intvl = 0$))}
\STATE \textit {Checkpoint: $i, \rho_i$ and $p^{(i)}, x^{(i)}$}
\ENDIF
\IF {(($i > 0$) and (recover))}
\STATE \textit{Recover: $A, M, i, \rho_i, p^{(i)}, x^{(i)}$}
\STATE Compute $r^{(i)} = b - Ax^{(i)}$
\ENDIF
\STATE $q^{(i)} = Ap^{(i)}$
\STATE $\alpha_i = \rho_i/{p^{(i)}}^Tq^{(i)}$
\STATE $x^{(i+1)} = x^{(i)} + \alpha_i p^{(i)}$
\STATE $r^{(i+1)} = r^{(i)} - \alpha_i q^{(i)}$
\STATE solve $Mz^{(i+1)} = r^{(i+1)}$
\STATE $\rho_{i+1} = r^{(i+1)}{}^Tz^{(i+1)}$
\STATE $\beta_i = \rho_{i+1}/\rho_{i}$
\STATE $p^{(i+1)} = z^{(i+1)} + \beta_i p^{(i)}$
\STATE check convergence; continue if necessary
\ENDFOR
\end{algorithmic}
\end{algorithm}

Based on this scheme, we can construct fault-tolerant iterative methods based on the checkpoint/recovery technique.
We use the preconditioned CG algorithm as an example, as shown in Algorithm 1.
This algorithm is one of the most commonly used iterative methods to solve sparse, symmetric, and positive-definite (SPD) linear systems.
It computes successive approximations to the solution (vector $x^{(i)}$), residuals corresponding to the approximate solutions (vector $r^{(i)}$), and search directions (vector $p^{(i)}$) used to update both the approximate solutions and the residuals.
Each iteration involves one sparse matrix-vector multiplication (line \textcolor{red}{10}), three vector updates (lines \textcolor{red}{12}, \textcolor{red}{13}, and \textcolor{red}{17}), and two vector inner products (lines \textcolor{red}{11} and \textcolor{red}{15}). We refer readers to \cite{barrett1994templates} for more details about CG method.  

For the CG algorithm, the matrix $A$, preconditioner $M$, and right-hand side vector $b$ are static variables.
The number of iterations $i$, the scalar $\rho$, the direction vector $p^{(i)}$, and the approximate solution vector $x^{(i)}$ are dynamic variables.
The residual vector $r^{(i)}$ is the recomputed variable, since we want to reduce checkpoint time and storage consumption.
Based on the checkpoint/recovery model for iterative methods discussed above, we perform checkpointing for $i$, $\rho$, $p^{(i)}$, and $x^{(i)}$ every $ckpt\_intvl$ iterations; 
and we perform recovering for $A$, $M$, $i$, $\rho$, $p^{(i)}$, and $x^{(i)}$ after a failure.

So far, we have constructed a fault-tolerant PCG solver with the checkpoint/recovery technique that has a strong resilience to failure-stop errors.
Based on this scheme, we now can construct the fault-tolerant algorithm for any iterative method as follows.
During the recovery, the first step is to recover a correct computational environment, such as an MPI environment.
It is usually achieved by performing a global restart of the execution.
Without loss of generality, we assume that the correct environment has been already recovered and that the recovered number of processors and tasks is the same as the previous failed one.

During the execution of iterative methods with checkpointing techniques, the overall checkpointing/restart cost is dominated by the dynamic variables instead of static variables. The reason is twofold. On the one hand, static variables are not involved in the checkpointing period but only the recovery step, while the optimal checkpointing frequency is generally considerably higher than the recovery frequency (i.e., failure rate). Suppose the mean time to interruption (MTTI) is 4 hours (i.e., 1 failure per 4 hours affecting the execution) and setting one checkpoint takes 18 seconds. Then the optimal checkpointing frequency is 5 checkpoints per hour according to Young's formula \cite{young}, which is 30 times as large as the failure rate. On the other hand, the static variables generally have comparable sizes with dynamic variables. Specifically, the static variables in the iterative methods are composed of the linear system matrix $A$, the preconditioner $M$, and right-hand side vector $b$. According to SuiteSparse Matrix Collection \cite{ufl} (formerly known as University of Florida Sparse Matrix Collection), the number of nonzeros (i.e., the data that needs to be stored) in matrix $A$ is usually of similar order to or a constant times (e.g., 1x${\sim}$10x) large than the dimension of dynamic vectors. 
For preconditioner $M$, it can be much more sparse than $A$. For example, the most commonly used preconditioning methods---\textit{block Jacobi} and \textit{incomplete LU factorization} (ILU)---need to store only the block diagonal matrix of $A$ and the matrix $L$, $U$ (where $A \approx LU$), respectively.
Therefore, the data size of static variables is usually the same order as or a constant times large than that of dynamic variables.
Taking these factors into account, we see that the overall checkpoint/restart overhead depends mainly on dynamic variables. Thus we focus mainly on reducing the checkpoint/recovery overhead of the dynamic variables in iterative methods by lossy compressors.
Note that when we build the lossy checkpointing performance model (Section \ref{sec:lossy-perf-model}) and perform the evaluation (Section \ref{sec:evaluation}), we take into account all the three types of variables instead of only dynamic variables.

\section{Lossy Checkpointing Scheme for Iterative Methods}
\label{sec:lossy}

In this section, we first analyze the expected overhead of checkpointing techniques for iterative methods. 
We prove that reducing the checkpointing overhead (e.g., by leveraging compression techniques) can significantly improve the overall performance, especially for future HPC systems.
This analysis motivates us to design an approach to reduce the checkpointing overhead.
Then, we propose our lossy checkpointing scheme that can be easily applied to iterative methods in numerical linear algebra.
We also present a new performance model for our lossy checkpointing scheme. Based on the model, we derive an upper bound for the number of extra iterations caused by lossy checkpoints against the reduced checkpointing overhead, to guarantee the performance improvement of the lossy checkpointing scheme.
We theoretically and empirically analyze the impact of lossy checkpointing on the convergence of iterative methods considering multiple types of iterative methods.

\subsection{Theoretical Analysis of Checkpointing Overhead for Iterative Methods}
If a failure happens, we restart the computation from the last checkpointed variables, as shown in Algorithm 2. 
This process is normally called \textit{rollback}. 
Rollback means that some previous computations need to be performed again.
Thus, the checkpointing frequency or time interval needs to be determined carefully.
Here the checkpointing interval means the mean time between two checkpoints.
On the one hand, a larger checkpointing interval means a longer rollback in case of failure, indicating more workload to be recomputed after the recovery;
on the other hand, a smaller checkpointing interval means more frequent checkpointing, leading to higher checkpointing overhead.
How to calculate the optimal checkpointing intervals has been studied for many years \cite{young, agbaria1999starfish}.
Our following analysis is based on the recovery pattern of iterative methods constructed by Langou et al. \cite{langou}. 


\begin{table}[tbp]
\centering
\caption{Notation for traditional checkpointing model}
\label{tab:notation}
\resizebox{\hsize}{!}{%
\begin{tabular}{|l|l|}
\hline
$T_{it}$  & Mean time of an iteration                                                                                       \\ \hline
$T_{ckp}$ & Mean time to perform a checkpoint                                                                             \\ \hline
$T_{rc}$  & \begin{tabular}[c]{@{}l@{}}Mean time to recover the application with the correct \\ environment and data from the last checkpoint\end{tabular}        \\ \hline

$T_{rb}$  & \begin{tabular}[c]{@{}l@{}}Mean time to perform a rollback of some redundant \\  computations\end{tabular}                                                                        \\ \hline
$T_f$     & Mean time to interruption                                                                                     \\ \hline
$T_{overhead}^{CR}$     & Mean time overhead of checkpoint/recovery                                                                                     \\ \hline
$\lambda$ & Failure rate, i.e, 1/$T_f$                                                                                      \\ \hline
$k$       & \begin{tabular}[c]{@{}l@{}}Checkpoint frequency - a checkpoint is performed\\ every $k$ iterations\end{tabular} \\ \hline
$N$       & Number of iterations to converge without failures                                                               \\ \hline
\end{tabular}
}
\end{table}

We use the notation in Table \ref{tab:notation} to analyze the expected fault tolerance overhead.
The overall execution time $T_t$ can be expressed as 
\begin{equation*}
\vspace{-0.5mm}
T_{t} = NT_{it} + T_{ckp}\frac{N}{k} + \frac{T_{t}}{T_{f}}(T_{rc} + T_{rb}).
\vspace{-0.5mm}
\end{equation*}

Without loss of generality, based on Young's formula \cite{young}, the optimal checkpointing interval should be chosen as 
\begin{equation}
\label{eq:young}
\vspace{-0.5mm}
k\cdot T_{it} = \sqrt{2T_{f}\cdot T_{ckp}}, 
\vspace{-0.5mm}
\end{equation}
and the expected mean time to perform a roll back, namely, $T_{rb}$, is $kT_{it}/2$. 
Thus,
\begin{align*}
\vspace{-1mm}
T_{t} &= NT_{it} + T_{ckp}\frac{T_{t}}{\sqrt{2T_{f}\cdot T_{ckp}}} + \frac{T_{t}}{T_{f}}(T_{rc} + \frac{\sqrt{2T_{f}\cdot T_{ckp}}}{2}) \\
	&= NT_{it} + T_{t}(\sqrt{\frac{2T_{ckp}}{T_f}} + \frac{T_{rc}}{T_f}) = NT_{it} + T_{t}(\sqrt{2\lambda T_{ckp}} + \lambda T_{rc}).
\vspace{-0.5mm}
\end{align*}

Similar to \cite{langou}, we therefore can get the expected overall execution time as 
\begin{equation}
\vspace{-0.5mm}
T_{t} = \frac{NT_{it}}{1 - \sqrt{2\lambda T_{ckp}} - \lambda T_{rc}},
\vspace{-0.3mm}
\end{equation}
and the fault tolerance overhead is 
\begin{equation}
\vspace{-0.3mm}
T_{overhead}^{CR} = T_{t} - NT_{it} = NT_{it}\cdot \frac{\sqrt{2\lambda T_{ckp}} + \lambda T_{rc}}{1 - \sqrt{2\lambda T_{ckp}} - \lambda T_{rc}},
\vspace{-0.3mm}
\end{equation}
where $NT_{it}$ is the basic productive execution time with $N$ iterations to converge. Note that in the paper, we use \textit{fault tolerance overhead} to refer to the performance overhead caused by checkpoints/recoveries and failure events, which is equal to the total running time taking away the basic productive execution time (i.e., $T_t$$-$$N$$T_{it}$).

We assume $T_{rc} \approx T_{ckp}$ without loss of generality. Then we can simplify the expected fault tolerance overhead as follows.
\begin{equation}
\label{eq:lossless-overhead}
\vspace{-0.5mm}
T_{overhead}^{CR} \approx NT_{it}\cdot \frac{\sqrt{2\lambda T_{ckp}} + \lambda T_{ckp}}{1 - \sqrt{2\lambda T_{ckp}} - \lambda T_{ckp}}
\vspace{-0.5mm}
\end{equation}
Moreover, we can calculate the ratio of the expected fault tolerance overhead to the basic productive execution time as Equation (\ref{eq:Toverhead}).
\begin{equation}
\label{eq:Toverhead}
\frac{T_{overhead}^{CR}}{NT_{it}} = \frac{\sqrt{2\lambda T_{ckp}} + \lambda T_{ckp}}{1 - \sqrt{2\lambda T_{ckp}} - \lambda T_{ckp}}
\end{equation}

\begin{figure}[t]
\centering
\includegraphics[scale=0.30]{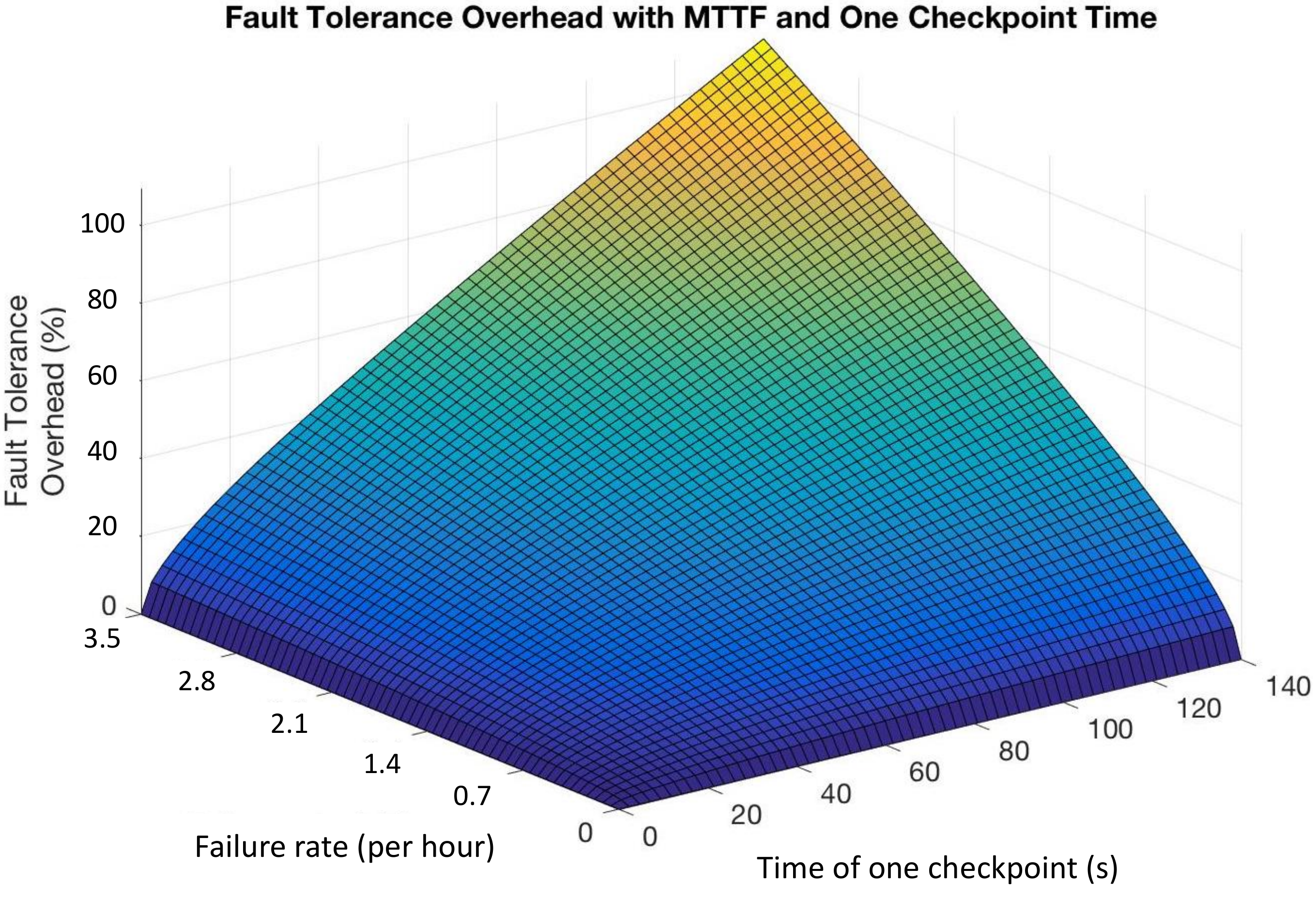}
\caption{Expected fault tolerance overhead with different failure rates and checkpoint time.}
\label{fig:f1}
\end{figure}

Now the expected fault tolerance overhead depends only on the failure rate $\lambda$ and time of one checkpoint $T_{ckp}$.
Based on this formula, we can plot the expected overhead of checkpoint/recovery based on different $\lambda$ and $T_{ckp}$, as shown in Figure \ref{fig:f1}.
We choose $\lambda$ from $0$ to $3.5$ failures per hour (i.e., MTTI from about $20$ minutes to infinity) and $T_{ckp}$ from $0$ to $140$ seconds. Note that the MTTI represents the expected period at which the application execution is interrupted.
Based on our experimental evaluation, checkpointing one dynamic vector $x$ once without compression takes about $120$ seconds with $2,048$ processes/cores on the Bebop cluster \cite{bebop} at Argonne National Laboratory.
In our experiment, the number of elements in the vector is set to $10^{10}$ (with $78.8$ GB double-precision floating-point data), which is the largest problem size that the three iterative methods (Jacobi, GMRES, and CG) can be run on the Bebop using 2,048 cores. 
We adopt the FTI library \cite{fti} with MPI-IO for checkpointing because of its high I/O efficiency confirmed in recent studies \cite{mpiio}. More details are presented in the experimental evaluation section.

Figure \ref{fig:f1} illustrates that the expected fault tolerance overhead can be as high as $40\%$ with $T_{ckp} = 120s$ if the MTTI is about hourly.
On future extreme-scale systems with millions of components, the failure rate may be higher, and the fault tolerance overhead issue could be more severe.
From Figure \ref{fig:f1}, we see that reducing the checkpointing time can significantly improve the overall performance of checkpoint/restart, especially under a higher error rate scenario.

\subsection{Lossy Checkpointing Scheme for Iterative Methods}
\label{sec:lossy-chkpt-scheme}
Our  lossy checkpointing scheme based on an iterative method has two key steps.
\begin{itemize}
\item Compress dynamic variables with lossy compressor before each checkpointing.
\item Decompress compressed dynamic variables after each recovering.
\end{itemize}

\begin{algorithm}
\caption{Fault-tolerant preconditioned conjugate gradient algorithm with lossy checkpointing technique}
\begin{flushleft}
\textbf{Input:} linear system matrix $A$, preconditioner $M$, and right-hand side vector $b$\\
\textbf{Output:} approximate solution $x$
\end{flushleft}
\begin{algorithmic}[1]
\STATE Initialization: same as line 1 in Algorithm 1
\FOR{$i = 0, 1, \cdots$}
\IF {(($i > 0$) and ($i\%ckpt\_intvl = 0$))}
\STATE \textbf{Compress: $x^{(i)}$ with lossy compressor}
\STATE \textit{Checkpoint: $i$ and compressed $x^{(i)}$}
\ENDIF
\IF {(($i > 0$) and (recover))}
\STATE \textit{Recover: $A, M, i$ and compressed $x^{(i)}$}
\STATE \textbf{Decompress:  $x^{(i)}$ with lossy compressor}
\STATE Compute $r^{(i)} = b - Ax^{(i)}$
\STATE Solve $Mz^{(i)} = r^{(i)}$
\STATE $p^{(i)} = z^{(i)}$
\STATE $\rho_i = {r^{(i)}}^Tz^{(i)}$
\ENDIF
\STATE Computation: same as lines 10--17 in Algorithm 1
\ENDFOR
\end{algorithmic}
\end{algorithm}

We still use the CG algorithm as an example, as shown in Algorithm 2, and the lossy checkpointing scheme can be applied to other iterative methods similarly.
Because of space limitations, we present only the lossy checkpointing part without the original computations in Algorithm 2. 
The lossy compression and decompression procedures are marked in bold.
We note that the CG algorithm maintains a series of orthogonality relations between the residual vectors $r$ and the direction vectors $p$.
Specifically, (1) $p^{(k)}$ and $Aq^{(j)}$, (2) $r^{(k)}$ and $p^{(j)}$, and (3) $r^{(k)}$ and $r^{(j)}$ are orthogonal to each other, where $j < k$.
However, these orthogonality relations may be broken after a recovery because of the errors introduced by lossy compression.
Unfortunately, the convergence rate of the CG algorithm (which is superlinear) is highly dependent on these orthogonality relations.
Hence, after a recovery from lossy checkpointing, the CG algorithm may lose the superlinear convergence rate, leading to a slow convergence \cite{sao2013self}.
To avoid this situation, we adopt a restarted scheme for the CG algorithm (\emph{restarted CG}) \cite{powell1977restart},
in which the computed approximate solution $x_i$ is periodically treated as a new guess. 
In this case, we need to checkpoint only the vector $x_i$ during the execution. The decompressed $x_i$ is used as a new initial guess, and a new series of orthogonal vectors is reconstructed for the execution, such that a superlinear convergence rate can be rebuilt after restart.

Some studies of iterative methods have proved that such a restarted version of iterative methods (i.e., restarting by treating the current approximate solution vector as a new initial guess) may bring important advantages \cite{powell1977restart,agullo2013towards,gmres}. 
On the one hand, it suffers from less time and space complexity compared to their classic counterparts. For example, in practice, GMRES  is often used to be restarted every a number of iterations (denoted by $k$) with the vector $x_k$ as a new initial guess; and it is often denoted by GMRES($k$).
Without the periodically restarting feature, the total time and space complexity of GMRES will both grow with an increasing rate of $N^2$ over the time step $N$.
By contrast, the time and space complexity of GMRES($k$) will be limited under a constant cost over the execution. 
On the other hand, some studies \cite{agullo2013towards,gmres} have proved that the \emph{restarted scheme} may not delay the convergence but even accelerate it, in that the periodically refreshed settings may enable the convergence to jump out of local search of the solution. 
In Section \ref{sec:impact}, we present more details regarding CG and GMRES with lossy checkpointing.
In the following discussion, we always use CG and GMRES to denote the restarted CG and GMRES, respectively, in the context of lossy checkpointing. For these restarted iterative methods, the only dynamic variable we need to checkpoint is the approximate solution vector $x$. 
As shown in Section \ref{sec:checkpoint}, however, even checkpointing one dynamic vector will still lead to a severe performance issue for current or future HPC systems.

Users can follow the below workflow to easily integrate our lossy checkpointing for iterative methods with existing HPC applications: (1) initialize application; (2) \textit{register variables external to the solver to checkpoint}; (3) start application's computations/iterations; (4) enter the solver's library; (5) \textit{register the solver's variables to checkpoint in the library}; (6) iterate the solver; (7) \textit{save or restore the application and solver's variables}; (8) continue to iterate the solver; (9) exit the solver's library; (10) continue application's computations/iterations (if needed); (11) exit application. Specifically, users can use the APIs, \texttt{Protect()} and \texttt{Snapshot()}, provided by our lossy checkpointing library to register and save/restore variables, as discussed in (2), (5), and (7).

\subsection{Performance Model of Lossy Checkpointing}
\label{sec:lossy-perf-model}

In this subsection, we build a performance model for the lossy checkpointing scheme, which is fundamental for analyzing the lossy checkpointing performance theoretically. Based on this performance model, we further derive a sufficient condition, an upper bound of the extra number of iterations caused by lossy data (i.e., Equation (\ref{eq:nbound})), for guaranteeing the performance improvement of the lossy checkpointing scheme. 
Building the performance model requires a few more parameters, as listed in Table 2.

\begin{table}[h]
\centering
\caption{Notations used in the lossy checkpointing performance model}
\resizebox{\hsize}{!}{%
\begin{tabular}{|l|l|}
\hline
$T_{comp}$        & Mean time of performing lossy compression                                                         \\ \hline
$T_{decomp}$      & Mean time of performing lossy decompression                                                       \\ \hline
$T_{ckp}^{trad}$  & Mean time of performing one traditional checkpoint                                                \\ \hline
$T_{ckp}^{lossy}$              & Mean time of performing a lossy checkpointing                                                     \\ \hline
$T_{overhead}^{lossyCR}$              & Time overhead of performing lossy checkpoint/recovery                                                     \\ \hline
$N'$              & Mean number of extra iterations caused by per lossy recovery                                                              \\ \hline
\end{tabular}
}
\end{table}

Since lossy compression introduces errors in the reconstructed dynamic variable(s), the solver may suffer from a delay to converge.
Suppose one recovery will cause extra $N'$ iterations to the convergence on average, then the total execution time can be rewritten as
\begin{equation}
\label{eq:lossy-ckpt-model}
T_{t} = NT_{it} + T_{ckp}^{lossy}\frac{N}{k} + \frac{T_{t}}{T_{f}}(N' T_{it} + T_{rc}^{lossy} + T_{rb}),
\end{equation}
because  lossy checkpointing needs to perform one decompression during each recovery,  lossy checkpointing needs to perform one compression during each checkpoint,
and each recovery will delay $N'$ iterations on average.
Note that $T_{ckp}^{lossy}$ and $T_{rc}^{lossy}$ include the compression time $T_{comp}$ and decompression time $T_{decomp}$, respectively. 
According to \cite{sz17, sz16, zfp}, $T_{comp}$ and $T_{decomp}$ are usually stable for a fixed compression accuracy.

Although the checkpointing/restarting time may differ across various iterations because of different data sizes due to various compression ratios, most  well-known iterative methods can converge quickly such that the value of each element in the approximate solution changes little in the following execution. Hence, the checkpointing data and its size will not change dramatically after several initial iterations, and $T_{ckp}^{lossy}$ and $T_{rc}^{lossy}$ can be assumed to be independent of iterations without loss of generality.

Similar to Section \ref{sec:checkpoint}, we can derive the expected execution time with lossy checkpointing as
\begin{equation*}
\vspace{-2mm}
T_{t} = \frac{NT_{it}}{1 - \sqrt{2\lambda T_{ckp}^{lossy}} - \lambda T_{rc}^{lossy} - \lambda N' T_{it}}
\end{equation*}
and the expected performance overhead of lossy checkpointing as
\begin{equation}
\label{eq:lossy-overhead-ori}
T_{overhead}^{lossyCR} = NT_{it}\cdot \frac{\sqrt{2\lambda T_{ckp}^{lossy}} + \lambda T_{rc}^{lossy} + \lambda N' T_{it}}{1 - \sqrt{2\lambda T_{ckp}^{lossy}} - \lambda T_{rc}^{lossy} - \lambda N' T_{it}}.
\end{equation}
Similarly, we can use $T_{ckp}^{lossy}$ to approximate $T_{rc}^{lossy}$ and simplify the performance overhead formula to 
\begin{equation}
\label{eq:lossy-overhead}
T_{overhead}^{lossyCR} \approx NT_{it}\cdot \frac{\sqrt{2\lambda T_{ckp}^{lossy}} + \lambda T_{ckp}^{lossy} + \lambda N' T_{it}}{1 - \sqrt{2\lambda T_{ckp}^{lossy}} - \lambda T_{ckp}^{lossy} - \lambda N' T_{it}}.
\end{equation}

Now, we can derive a sufficient condition for iterative methods such that the lossy checkpointing scheme  with a lossy compressor is able to obtain a performance gain over the traditional checkpointing scheme without lossy compression techniques.
\begin{theorem}
Denote $\lambda$ and $T_{it}$ by the expected failure rate and expected execution time of an iteration, respectively.
The lossy checkpointing scheme will improve the execution performance for an iterative method as long as the following inequality holds. 
\begin{equation}
\label{eq:nbound}
\begin{array}{l}
N' \leq 
(f(T_{ckp}^{trad},\lambda) - f(T_{ckp}^{trad},\lambda)) / (\lambda T_{it}),\\
where \hspace{1mm} f(t, \lambda) = \sqrt{2\lambda t} + \lambda t
\end{array}
\end{equation}
\end{theorem}
\begin{proof}
To have the lossy checkpointing overhead be lower than that of  traditional checkpointing, we make Equation (\ref{eq:lossy-overhead}) smaller than Equation (\ref{eq:lossless-overhead}):
\begin{equation*}
\resizebox{.9\hsize}{!}{%
$\frac{\sqrt{2\lambda T_{ckp}^{lossy}} + \lambda T_{ckp}^{lossy} + \lambda N' T_{it}}{1 - \sqrt{2\lambda T_{ckp}^{lossy}} - \lambda T_{ckp}^{lossy} - \lambda N' T_{it}} \leq
\frac{\sqrt{2\lambda T_{ckp}^{trad}} + \lambda T_{ckp}^{trad}}{1 - \sqrt{2\lambda T_{ckp}^{trad}} - \lambda T_{ckp}^{trad}}$.
}
\end{equation*}

Further simplifying this inequality, we can get the following formula with respect to the upper bound of $N'$.
\begin{equation*}
\resizebox{.76\hsize}{!}{%
$N' \leq \frac{(\sqrt{2\lambda T_{ckp}^{trad}} + \lambda T_{ckp}^{trad}) - (\sqrt{2\lambda T_{ckp}^{lossy}} + \lambda T_{ckp}^{lossy})}{\lambda T_{it}}$
}
\end{equation*}
Rewriting this inequality with $f(t, \lambda) = \sqrt{2\lambda t} + \lambda t$ will lead to Equation (\ref{eq:nbound}).
\end{proof}


We give an example to explain how to use Theorem 1 in practice. 
Based on our experiments running GMRES on the Bebop cluster with 2,048 cores, we noted that the lossy compression technique can reduce the checkpointing time $T_{ckp}$ from $120$ seconds to $25$ seconds for GMRES with a checkpoint of about $80$ GB data (details are presented later in Figure \ref{fig:f3}). 
Suppose the MTTI of a system is one hour (i.e., $\lambda = 1/3600$) and that the GMRES algorithm requires $5,875$ iterations with a total of $7,160$ seconds to converge. Then the mean time of one iteration, namely, $T_{it}$, is about $1.2$ seconds.
We can derive the maximum acceptable number of extra iterations to be 500 based on Equation (\ref{eq:nbound}). Hence,  using a lossy checkpointing scheme is worthwhile if one recovery (with compression error introduced by lossy checkpointing) causes extra $500$ or fewer iterations (about 9\% of total iterations) to converge,

\subsection{Impact Analysis of Lossy Checkpointing on Iterative Methods}
\label{sec:impact}

In this subsection, we analyze the impact of lossy checkpointing on iterative methods, including stationary iterative methods, GMRES, and CG.
Based on our analysis, we conclude that our lossy checkpointing technique can be applied to most of the iterative methods in numerical linear algebra for reducing the fault tolerance overhead.

\subsubsection{Stationary Iterative Methods}
\label{sec:stat-bound}
We analyze the impact of lossy checkpointing on the convergence of four representative iterative methods: Jacobi, Gauss-Seidel, successive overrelaxation, and symmetric successive overrelaxation.
The stationary iterative methods can be expressed in the following simple form,
\begin{equation*}
x^{(i)} = G x^{(i-1)} + c,
\end{equation*}
where $G$ and $c$ are a constant matrix and a vector, respectively.

Let $R$ denote the spectral radius of matrix $G$, which is the largest eigenvalue of the matrix $G$.
The convergence rate of a stationary iterative method is determined by its value.
Specifically, let $x^*$  denote the exact solution of the linear system:
\begin{equation*}
||x^{(i)} - x^* || \approx R^i \cdot || x^{(0)} - x^*||.
\end{equation*}
Since the initial guess $x^{(0)}$ could be any vector and it is set to $0$ in general, we have 
\begin{equation}
\label{eq:jacobi-rate}
||x^{(i)} - x^* || \approx R^i \cdot || x^*||.
\end{equation}

Suppose the stationary methods encounter a failure and restart at the $t^{th}$ iteration, and we denote that the lossy compression introduces an error vector $e$ to $x^{(t)}$ by following relative error bound without loss of generality. Here the relative error bound means $|x^{(t)}_i - x'^{(t)}_i|\leq eb \cdot |x^{(t)}_i|$ for $1\leq i \leq n$, where $eb$ is the relative error bound, $x^{(t)}_i$ is the $i$th element of vector $x^{(t)}$, and $n$ is the vector length. 
The computation will start from $x^{(t)}+e$ (denoted by $x'^{(t)}$) instead of $x^{(t)}$.
We derive the following theorem to obtain the range of the \textbf{expected upper bound of the number of extra iterations} for the stationary iterative methods.
\begin{theorem}
Based on the convergence rate (Equation (\ref{eq:jacobi-rate})), the expected upper bound of the number of extra iterations for the stationary iterative methods falls into the interval $[\mathit{\frac{N+1}{2}-\log_R(R^{\frac{N+1}{2}}+eb)}$, $\mathit{N-\log_R(R^N+eb)}]$, where $eb$ is a constant relative error bound and $R$ and $N$ remain the same definitions as in the earlier discussion.
\end{theorem}

\begin{proof}
Based on the definition of the relative error bound $eb$, we have $||e|| \leq eb \cdot ||x^{(t)}||$. 
Then, we can get 

\begin{align}
\vspace{-1mm}
\label{eq:jacobi-1}
||x'^{(t)} - x^* || & = ||x^{(t)} + e - x^*|| \leq ||x^{(t)} - x^*|| + ||e|| \\
& \leq R^{t} \cdot ||x^*|| + eb \cdot ||x^{(t)}|| .\nonumber
\end{align}
After another $m$ iterations from erroneous vector $x'^{(t)}$, we have
\begin{equation*}
||x^{(t+m)} - x^*|| \approx R^m \cdot ||x'^{(t)} - x^*|| .
\end{equation*}
Then, based on Equation (\ref{eq:jacobi-1}), we derive the following inequality.
\begin{equation}
\label{eq:jacobi-2}
||x^{(t+m)} - x^*|| \leq R^m (R^{t} \cdot ||x^*|| + eb \cdot ||x^{(t)}||) 
\end{equation} 

Let us consider how to choose an $m$ to ensure $||x^{(t+m)} - x^*|| \leq ||x^{(t)} - x^*||$, so that the residual norm between the approximate solution and exact solution will return to the previous value after  $m$ steps.
Based on Equation 
(\ref{eq:jacobi-2}), if we assure 
\begin{equation}
\label{eq:jacobi-3}
R^m (R^{t} \cdot ||x^*|| + eb \cdot ||x^{(t)}||)  \leq ||x^{(t)} - x^*||,
\end{equation}
then $||x^{(t+m)} - x^*|| \leq ||x^{(t)} - x^*||$ will hold.
Also, based on Equation (\ref{eq:jacobi-rate}), $||x^{(t)} - x^*|| \approx R^t ||x^*||$, Equation (\ref{eq:jacobi-3}) is equivalent to 
\begin{equation*}
R^m (R^{t} \cdot ||x^*|| + eb \cdot ||x^{(t)}||)  \leq R^t ||x^*||.
\end{equation*}
Therefore, $ m \geq \log_{R}\frac{R^t \cdot ||x^*||}{ (R^{t} \cdot ||x^*|| + eb \cdot ||x_{t}||) }$.


Without loss of generality, $||x^{(t)}||$ is likely close to $||x^*||$ after running with a few initial iterations, so we have the following approximation:
\begin{align*}
\log_{R}\frac{R^t \cdot ||x^*||}{ (R^{t} \cdot ||x^*|| + eb \cdot ||x^{(t)}||) } & \approx \log_{R}\frac{R^t}{ (R^{t} + eb) }.
\end{align*}
As a result of these inequalities, as long as $m \geq \log_{R}\frac{R^t}{ (R^{t} + eb) }$, we will assure $||x^{(t+m)} - x^*|| \leq ||x^{(t)} - x^*||$.
In other words, the stationary iterative methods need to take extra $\log_{R}\frac{R^t}{ (R^{t} + eb)}$ iterations \textit{at most} for convergence to the same accuracy.
To conclude, if the stationary methods restart at the $t$th iteration with relative error bound $eb$, the upper bound of extra iterations $N'$ is $t - \log_R(R^t + eb)$.

We now can calculate the range of the expected upper bound of extra iterations for the stationary iterative methods based on the monotonicity and convexity of $t - \log_R(R^t + eb)$ and Jensen inequality. Because of space limitations, we omit the details here.
\end{proof}



\subsubsection{GMRES}
\label{sec:gmres-bound}
The generalized minimum residual method proposed by Saad and Schultsz \cite{gmres} is a Krylov subspace method for solving a large, sparse linear system with no constraint on the coefficient system matrix, especially for solving nonsymmetric systems.
The method minimizes the norm of residual vector over a Krylov subspace at every iteration.
Considering the cost growth of GMRES, it is often executed with the restarting scheme.
In the following discussion, we use GMRES and restarted GMRES interchangeably.
Although GMRES has a good ability to resist silent data corruption \cite{ftgmres}, protecting GMRES against fail-stop errors still has to rely on a checkpointing technique.

Unlike the stationary iterative methods, analyzing the extra convergence steps for nonstationary methods is difficult in theory.
However, we propose an adaptive scheme to determine the error bound for GMRES with lossy checkpointing as follows.
\begin{theorem}
\label{th:gmres}
For the GMRES method, after a restart with lossy checkpointing, the new residual norm is controlled close to or at least on the same order as the previous residual if the relative error bound $eb$ is set to O($||r^{(t)}||/||b||$).
\end{theorem}
\begin{proof}
Similar to Equation (\ref{eq:jacobi-1}), we have the following. 
\begin{align}
||r'^{(t)}|| & = || b - A x'^{(t)}|| = ||b - Ax^{(t)} + A(x^{(t)} - x'^{(t)}) || \nonumber \\
& \leq ||r^{(t)}|| + ||Ae|| \leq ||r^{(t)}|| + eb \nonumber \cdot ||Ax^{(t)}||  \\
& = ||r^{(t)}|| + eb \cdot ||b - r^{(t)}|| \leq (1+eb)||r^{(t)}|| + eb \cdot ||b|| \nonumber \\
& \approx ||r^{(t)}|| + eb \cdot ||b|| 
\label{eq:residual-norm-error}
\end{align}
If $eb$ is set to $O$($||r^{(t)}||/||b||$), then $eb \cdot ||b||$ is O($||r^{(t)}||$); hence, $||r^{(t)}|| + eb \cdot ||b||$ is O($||r^{(t)}||$), which means that the new residual norm $||r'^{(t)}||$ will be of the same order as the previous residual norm $||r^{(t)}||$ based on Equation (\ref{eq:residual-norm-error}).
\end{proof}
Thanks to error-bounded compressors such as SZ and ZFP, one can easily control the distortion of data within $eb \cdot ||x^{(t)}||$. Theorem \ref{th:gmres} indicates that the convergence rate of GMRES will not degrade if the distortion of lossy checkpointing data follows a relative error bound $||r^{(t)}||/||b||$, where $t$ is the current iteration number. 

Now we can get a reasonable expected number of extra iterations for GMRES. As presented in Langou et al.'s study \cite{langou}, if it is the same order of residual norm with which the restarted GMRES forms a new approximate solution, GMRES will converge to the same accuracy with no delay or even exhibit an accelerated convergence sometimes.
The key reason is that the GMRES is easy to stagnate in its practical execution. If a failure occurs during the stagnation, the alternated recovered data can form a new approximate solution with different spectral properties, which may help GMRES jump out of the stagnation. 
This phenomenon has been theoretically and empirically observed and proved by Langou et al. \cite{langou}. Considering such a feature, the restarted GMRES with our proposed lossy checkpointing can converge without any delay based on the compression error bound suggested by Theorem \ref{th:gmres} with an ensured, controlled residual norm. 
Thus, we can set the expected $N'$ of GMRES with lossy checkpointing to $0$.
As a result, our lossy checkpointing is highly suitable for the restarted GMRES.


\subsubsection{Conjugate Gradient}
\label{sec:cg-bound}
The conjugate gradient method is usually used in non-restarted style and has a superlinear convergence rate. 
As discussed in Section \ref{sec:lossy-chkpt-scheme}, however, we adopt restarted CG with lossy checkpointing.
After a restart, it has to re-establish a new Krylov subspace based on the new initial guess. In our case, the new initial guess is the recovered vector $x'^{t}$ (decompressed by lossy compressor).
This process can lead to a delay of convergence to some extent.
Unlike the GMRES method, even if we can ensure that $eb \cdot ||x^{(t)}||$ is the same order as $||x^{(t)} - x^*||$, shown in Equation (\ref{eq:residual-norm-error}), the extra convergence steps for CG exhibit a property of randomness. Thus, for the CG method, we turn from theoretical analysis to an empirical evaluation for $N'$. 

For each experiment, we randomly select an iteration to compress the approximate solution vector, decompress it to continue the computations, and then count the number of extra iterations. We evaluate the average extra iterations with different relative error bounds, as shown in Figure \ref{fig:f5}. The figure shows that the errors introduced by lossy checkpointing may delay the convergence of CG method to a certain extent. Based on our evaluation, the average extra iterations varies from 10\% to 25\% with different error bounds.

\begin{figure}[]
\centering
\includegraphics[scale=0.50]{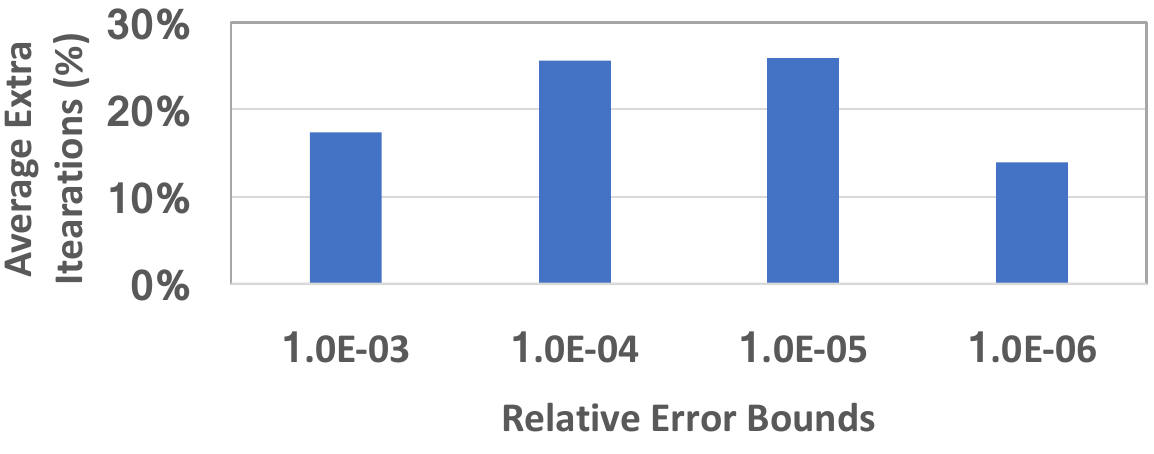}
\caption{Average extra iterations of CG method per lossy recovery with different error bounds.}
\label{fig:f5}
\end{figure}

\subsubsection{Reproducibility with Lossy Checkpointing}
Based on our experiments, iterative methods with our lossy checkpointing can still converge to a solution that satisfies the user-set accuracy. Moreover, the variance of the solution is much smaller than the user-set convergence tolerance threshold. Hence, our lossy checkpointing has an impact on bit-level reproducibility but only has a negligible impact on tolerance-based reproducibility of iterative methods and outer applications.
\section{Performance Evaluation}
\label{sec:evaluation}

In this section, we evaluate our proposed lossy checkpointing technique for iterative methods and compare it with traditional checkpointing and lossless checkpointing. 

\subsection{Experimental Setting}
We conduct our evaluation using 2,048 cores (i.e., 64 nodes, each with two Intel Xeon E5-2695 v4 processors and 128 GB memory, and each processor with 16 cores) from the Bebop cluster \cite{bebop} at Argonne National Laboratory.
Its I/O and storage systems are typical of high-end supercomputer facilities.

We implement our lossy checkpointing technique based on the FTI checkpointing library (v0.9.5) \cite{fti} and SZ lossy compression library (v1.4.12) \cite{sz17}.
The code is available in \cite{fti-sz}.
We use the MPI-IO mode \cite{mpiio} in FTI to write the checkpointing data to the parallel file system.
For the lossy compression, compared with other lossy compressors (such as ZFP \cite{zfp} and Tucker decomposition \cite{tucker}), SZ has a better performance for 1D data sets, as demonstrated in \cite{sz17, sz16}.
Most dynamic variables in lossy checkpointing are 1D vectors; hence, in this paper, we select SZ as our lossy compression approach.
We use a reasonable relative error bound of $10^{-4}$ \cite{sz16, sz17} for Jacobi and CG and set the relative error bound suggested by Theorem 3 for GMRES.
We choose the Gzip \cite{gzip} lossless compressor to represent the state-of-the-art lossless compression for comparison.
We call the checkpointing without compression as ``traditional checkpointing'' and the checkpointing with lossless compression as ``lossless checkpointing'' in order to correspond to lossy checkpointing.

We evaluate our proposed lossy checkpointing technique for the iterative methods implemented in PETSc (v3.8) \cite{petsc}.
We adopt its default preconditioner (block Jacobi with ILU/IC) and use the relative convergence tolerance \footnote{relative decrease in the (possibly preconditioned) residual norm with the default value of $10^{-5}$ in PETSc.} (denoted by $rtol$) of $1\mathrm{e}{-4}$, $7\mathrm{e}{-5}$, and $1\mathrm{e}{-7}$ for Jacobi, GMRES, and CG, respectively.
For GMRES, we use PETSc's recommended setting $30$ as its restarted step (i.e., GMRES(30)).

For demonstration purposes, we focus on solving the following sparse linear system (arising from discretizing a 3D Poisson's equation):
\begin{equation}
\label{eq:linear-system}
A_{n^3 \times n^3} x_{n^3 \times 1} = b_{n^3 \times 1},
\end{equation}
where
{
\begin{equation*}
\resizebox{.85\hsize}{!}{%
$A_{n^3 \times n^3} = 
\begin{pmatrix}
 M_{n^2 \times n^2}& I_{n^2 \times n^2}  &  &  & \\ 
 I_{n^2 \times n^2}& M_{n^2 \times n^2} &  I_{n^2 \times n^2}  &  & \\ 
 & \ddots & \ddots & \ddots & \\ 
 &  & I_{n^2 \times n^2} &  M_{n^2 \times n^2} &I_{n^2 \times n^2} \\ 
 &  &  & I_{n^2 \times n^2} &  M_{n^2 \times n^2}
\end{pmatrix},$}
\end{equation*}
\begin{equation*}
\resizebox{.7\hsize}{!}{%
$M_{n^2 \times n^2}  = 
\begin{pmatrix}
 T_{n \times n}& I_{n \times n}  &  &  & \\ 
 I_{n \times n}& T_{n \times n} &  I_{n \times n}  &  & \\ 
 & \ddots & \ddots & \ddots & \\ 
 &  & I_{n \times n} &  T_{n \times n} &I_{n \times n} \\ 
 &  &  & I_{n \times n} &  T_{n \times n}
\end{pmatrix},$}
\end{equation*}
\begin{equation*}
\resizebox{.5\hsize}{!}{%
$T_{n \times n}  = 
\begin{pmatrix}
  -6 & 1  &  &  & \\ 
 1 & -6 &  1  &  & \\ 
 & \ddots & \ddots & \ddots & \\ 
 &  & 1 &  -6 & 1 \\ 
 &  &  & 1 &  -6
\end{pmatrix},$}
\end{equation*}
}
so that we can increase the problem size as the scale increases.

Note that all stationary methods are similar to each other. Hence, without loss of generality, we focus our experiments for stationary iterative methods on the Jacobi method.
For nonstationary methods, we note that the sparse matrix $A_{n^3\times n^3}$ is symmetric and positive definite; hence, it can be used to test both CG and GMRES. 

\begin{table}[]
\centering
\caption{Problem sizes and average checkpoint sizes with different iterative methods and number of processes on Bebop}
\label{tab:size}
\begin{adjustbox}{max width=\columnwidth}
\begin{tabular}{|c|c|c|c|c|c|c|l|c|c|l|}
\hline
\multirow{3}{*}{\textbf{\begin{tabular}[c]{@{}c@{}}Num.\\ of\\ Proc.\end{tabular}}} & \multirow{3}{*}{\textbf{\begin{tabular}[c]{@{}c@{}}Problem\\ Size\end{tabular}}} & \multicolumn{9}{c|}{\textbf{Checkpoint Size Per Proc (MB)}}                                                                                                                                                                                                                                      \\ \cline{3-11} 
                                                                                  &                                                                                  & \multicolumn{3}{c|}{\textbf{\begin{tabular}[c]{@{}c@{}}Traditional\\ Checkpointing\end{tabular}}} & \multicolumn{3}{c|}{\textbf{\begin{tabular}[c]{@{}c@{}}Lossless\\ Checkpointing\end{tabular}}} & \multicolumn{3}{c|}{\textbf{\begin{tabular}[c]{@{}c@{}}Lossy\\ Checkpointing\end{tabular}}} \\ \cline{3-11} 
                                                                                  &                                                                                  & Jacobi                           & GMRES                          & CG                            & \multicolumn{1}{l|}{Jacobi}                   & GMRES                  & CG                    & \multicolumn{1}{l|}{Jacobi}                  & GMRES                 & CG                   \\ \hline
256                                                                               & $1088^3$                                                                         & 38.4                             & 38.4                           & 76.8                          & 5.99                                          & 34.6                   & 69.5                  & 1.33                                         & 1.23                  & 1.69                 \\ \hline
512                                                                               & $1368^3$                                                                         & 38.2                             & 38.2                           & 76.4                          & 5.96                                          & 34.0                   & 71.2                  & 1.35                                         & 1.13                  & 1.58                 \\ \hline
768                                                                               & $1568^3$                                                                         & 38.3                             & 38.3                           & 76.6                          & 5.98                                          & 34.1                   & 73.6                  & 1.37                                         & 1.21                  & 1.47                 \\ \hline
1024                                                                              & $1728^3$                                                                         & 38.4                             & 38.4                           & 76.8                          & 5.99                                          & 34.0                   & 69.4                  & 1.28                                         & 1.18                  & 1.49                 \\ \hline
1280                                                                              & $1856^3$                                                                         & 39.9                             & 39.9                           & 79.8                          & 6.24                                          & 33.6                   & 69.1                  & 1.33                                         & 1.19                  & 1.46                 \\ \hline
1536                                                                              & $1968^3$                                                                         & 39.7                             & 39.7                           & 79.4                          & 6.20                                          & 33.1                   & 69.2                  & 1.23                                         & 1.17                  & 1.42                 \\ \hline
1792                                                                              & $2064^3$                                                                         & 39.3                             & 39.3                           & 78.6                          & 6.13                                          & 32.8                   & 70.7                  & 1.30                                         & 1.17                  & 1.35                 \\ \hline
2048                                                                              & $2160^3$                                                                         & 39.4                             & 39.4                           & 78.8                          & 6.15                                          & 32.7                   & 67.9                  & 1.16                                         & 1.16                  & 1.33                 \\ \hline
\end{tabular}
\end{adjustbox}
\end{table}

In this paper, we focus mainly on the weak-scaling study for performance evaluation. 
We choose the largest problem size that can be held in memory by using 2,048 cores (i.e., 64 nodes) for GMRES(30), as shown in Table \ref{tab:size}.
For consistency, we also adopt these sizes for the Jacobi method and CG.
Table \ref{tab:size} also shows the corresponding checkpointing sizes per process with different scales (from $256$ to $2,048$ processes) and different checkpointing solutions.

\subsection{Evaluation of Iterative Methods with Large-scale Sparse Matrix from SuiteSparse}
Before evaluating the lossy checkpointing for iterative methods, we first evaluate the productive execution time of iterative methods with the largest symmetric indefinite sparse matrix (i.e., KKT240 with around 28 million linear equations) in SuiteSparse \cite{ufl} using 4,096 processes/cores on the Bebeop cluster at Argonne, as shown in Figure \ref{fig:f10}. Symmetric indefinite KKT matrices are generated from a nonlinear programming problems for a 3D PDE-constrained optimization problem \cite{nlpkkt}. We refer readers to \cite{nlpkkt240} for more details of the matrices. Note that we use GMRES for demonstration purpose, since it is much faster than Jacobi and CG cannot handle indefinite matrix. We test all the preconditioners listed in the PETSc's website \cite{petsc-linear-solver} and choose the best one (i.e., Jacobi preconditioner). We use the relative convergence tolerance of $1\mathrm{e}{-6}$. Figure \ref{fig:f10} shows that the average productive execution time for solving KKT240 once with GMRES needs to take more than one hour with 4,096 processes. Moreover, we note that the dimensions of the matrices collected by SuiteSparse grow exponentially with years \cite{ufl}. Therefore, it will be more common to spend hours to days running iterative methods with a large number of ranks in parallel for very large-scale sparse linear systems; on the other hand, the mean time between failures for petascale supercomputers could be hourly or even less than one hour, as demonstrated by a recent study \cite{liu2018large} based on a three-year statistic of Sunway TaihuLight supercomputer \cite{fu2016sunway}. These results demonstrate that checkpointing during iterative methods is very important for the future HPC applications and exascale systems.

\begin{figure}[]
\centering
\includegraphics[scale=0.58]{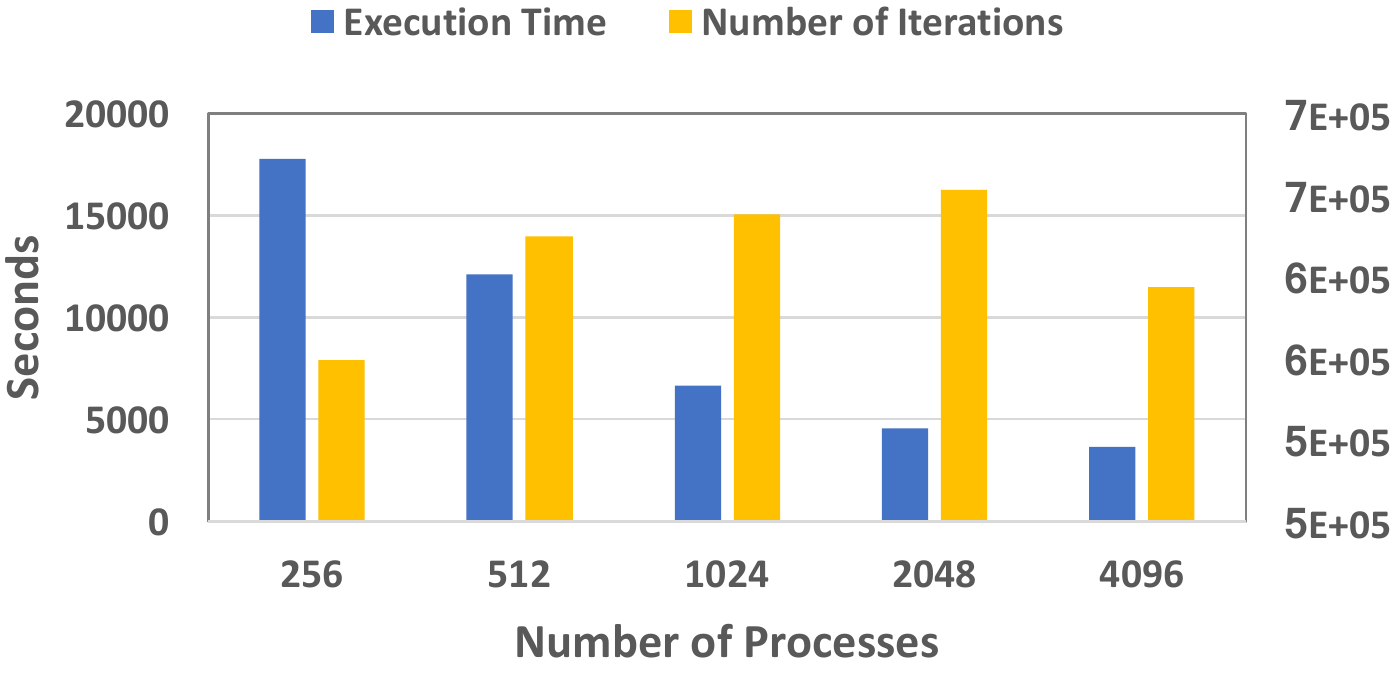}
\caption{Average productive execution times and numbers of convergence iterations for solving matrix KKT240 once using GMRES and Jacobi preconditioner with different number of processes on Bebop.}
\label{fig:f10}
\end{figure}

\subsection{Theoretical Performance Investigation}
We next perform the experiments with three checkpointing solutions under a fixed checkpoint frequency.
The objective is to obtain the mean size and time of one checkpoint/recovery across different iterations from beginning to end for the three solutions.
We set the checkpointing/recovering frequency to six times per hour and run each experiment for five times to ensure that the checkpoints/recoveries can cover the entire iteration.
We calculate the average size and time of one checkpoint/recovery with different scales.
We present the average checkpoint/recovery sizes for Jacobi, GMRES, and CG in Table \ref{tab:size}.
We present the average checkpoint/recovery time with different checkpointing solutions for Jacobi, GMRES, and CG in Figures \ref{fig:f7}, \ref{fig:f3}, and \ref{fig:f9}, respectively. 

\begin{figure}[]
\centering
\includegraphics[scale=0.68]{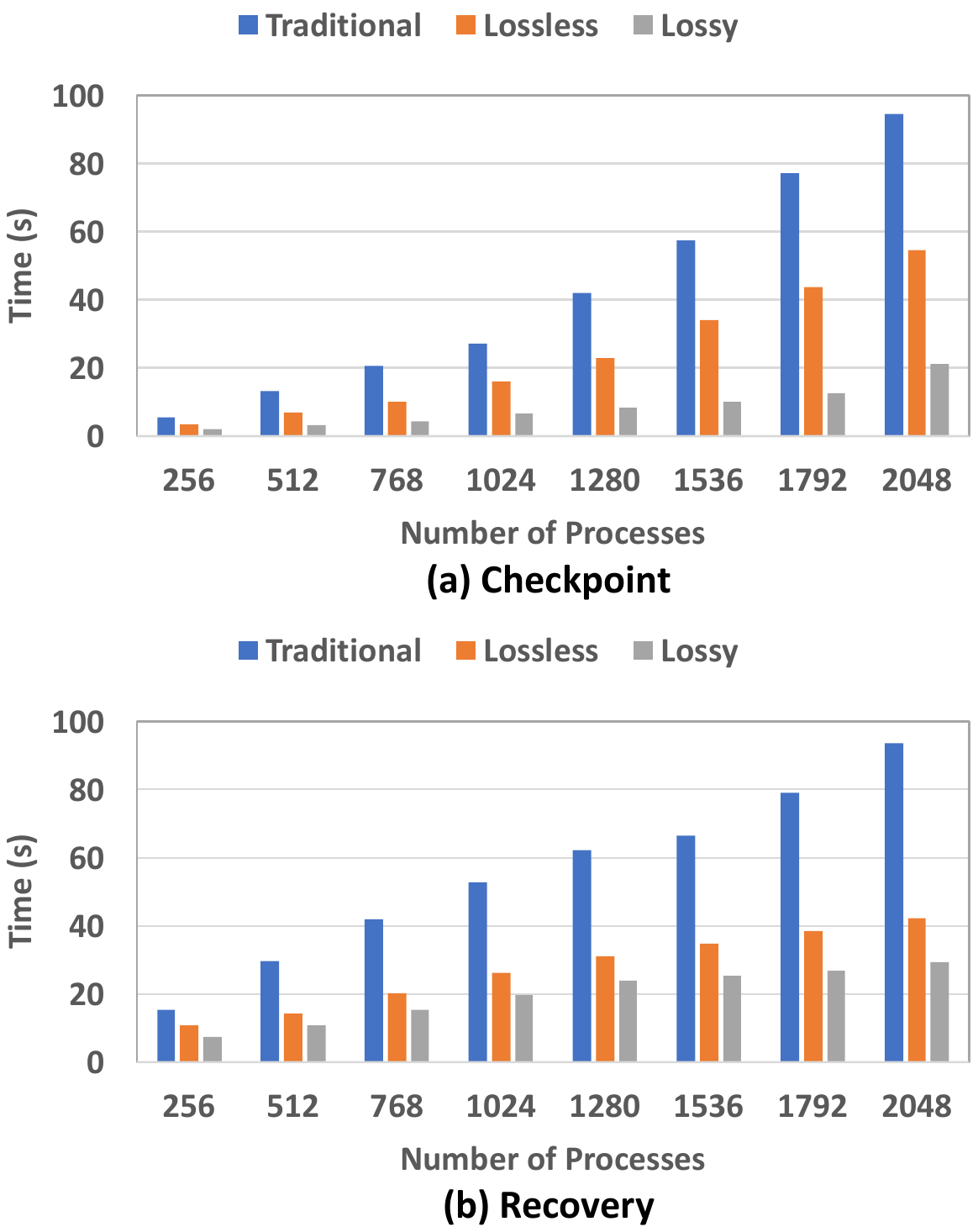}
\caption{Average time of one checkpoint and recovery for Jacobi with different checkpointing techniques on Bebop.}
\label{fig:f7}
\end{figure}

Table \ref{tab:size} illustrates that lossy compression can significantly reduce the checkpointing size compared with traditional and lossless checkpointing.
Specifically, SZ lossy compression can reduce the checkpointing size to about $1/20 \sim 1/60$, whereas the lossless compression can  achieve a compression ratio only up to about $6$. 
Consequently, the checkpoint/recovery time can be  reduced significantly for Jacobi, GMRES, and CG compared with the other two solutions, as shown in Figure \ref{fig:f7}, \ref{fig:f3}, and \ref{fig:f9}.
Comparing the three figures, we can see that the lossy checkpointing reduce checkpoint/recovery time more significantly for CG than for Jacob and GMRES.
The reason is that the traditional and lossless checkpointing methods need to checkpoint/recover two vectors ($x$ and $p$) for CG (as shown in line \textcolor{red}{4} in Algorithm 1) \cite{chen2011algorithm,chen2013online}, in that reinitializing $p$ based $x$ will lead to unknown delays (extra iterations).
However, we have investigated the impact of lossy checkpoints on extra iterations of restarted CG, thus only the vector $x$ needs to be checkpointed/recovered in our lossy checkpointing scheme. In addition, it is also observed that the checkpointing and recovery overhead both increase approximately linearly with the number of processes, because of linear increasing of the total checkpointing data size and the constant I/O bandwidth. In fact, such an I/O time increase is inevitable for any PFS considering the limited I/O bandwidth. 

\begin{figure}[]
\centering
\includegraphics[scale=0.69]{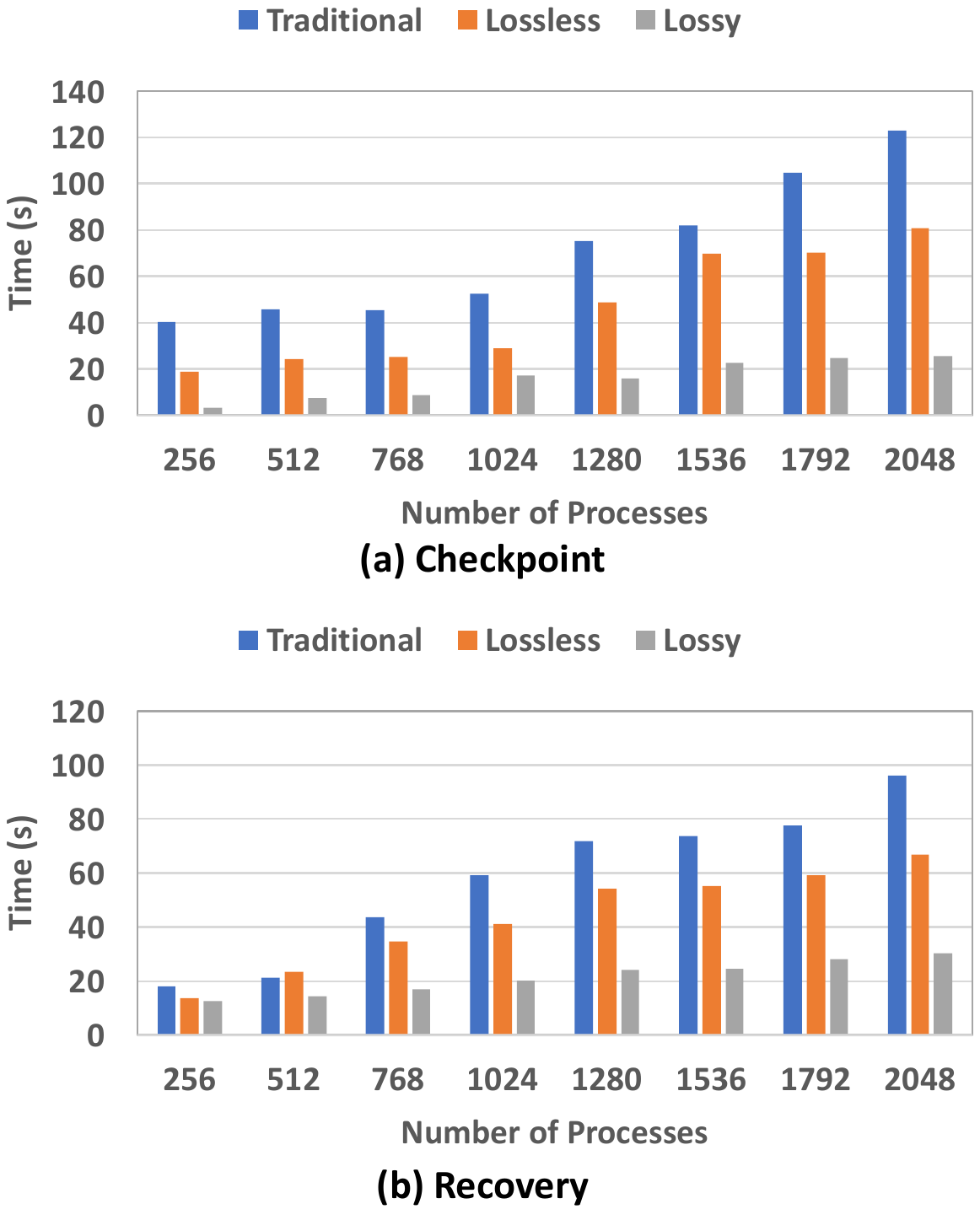}
\caption{Average time of one checkpoint and recovery for GMRES with different checkpointing techniques on Bebop.}
\label{fig:f3}
\end{figure}

We include the compression and decompression time in the checkpoint/recovery time. 
The recovery time also contains the time to reconstruct static variables, including matrix $A$, preconditioner $M$, and right-hand side vector $b$. 
As shown in \cite{sz17}, because of no communication in parallel compression and decompression, the efficiency of parallel compression can stay at $90\%$,
and the compression and decompression speed can reach  $80$ and $180$ GB/s with 1,024 cores, respectively.
Therefore, the compression and decompression take only a small portion of time in the checkpoint/recovery.
Specifically, compressing and decompressing the $78.8$ GB of checkpointing data take only about $0.5$ seconds and $0.2$ seconds, respectively.
Such cost is nearly negligible compared with the average checkpoint/recovery time. 
Note that the time spent on I/O will increase with the number of processors, because of the inevitable bottleneck of the bandwidth when writing/reading data by many processes simultaneously (even with parallel I/O).
By contrast, parallel compression/decompression time increases little with the number of processors, which means the performance gains by lossy checkpointing will increase with scales.

\begin{figure}[]
\centering
\includegraphics[scale=0.69]{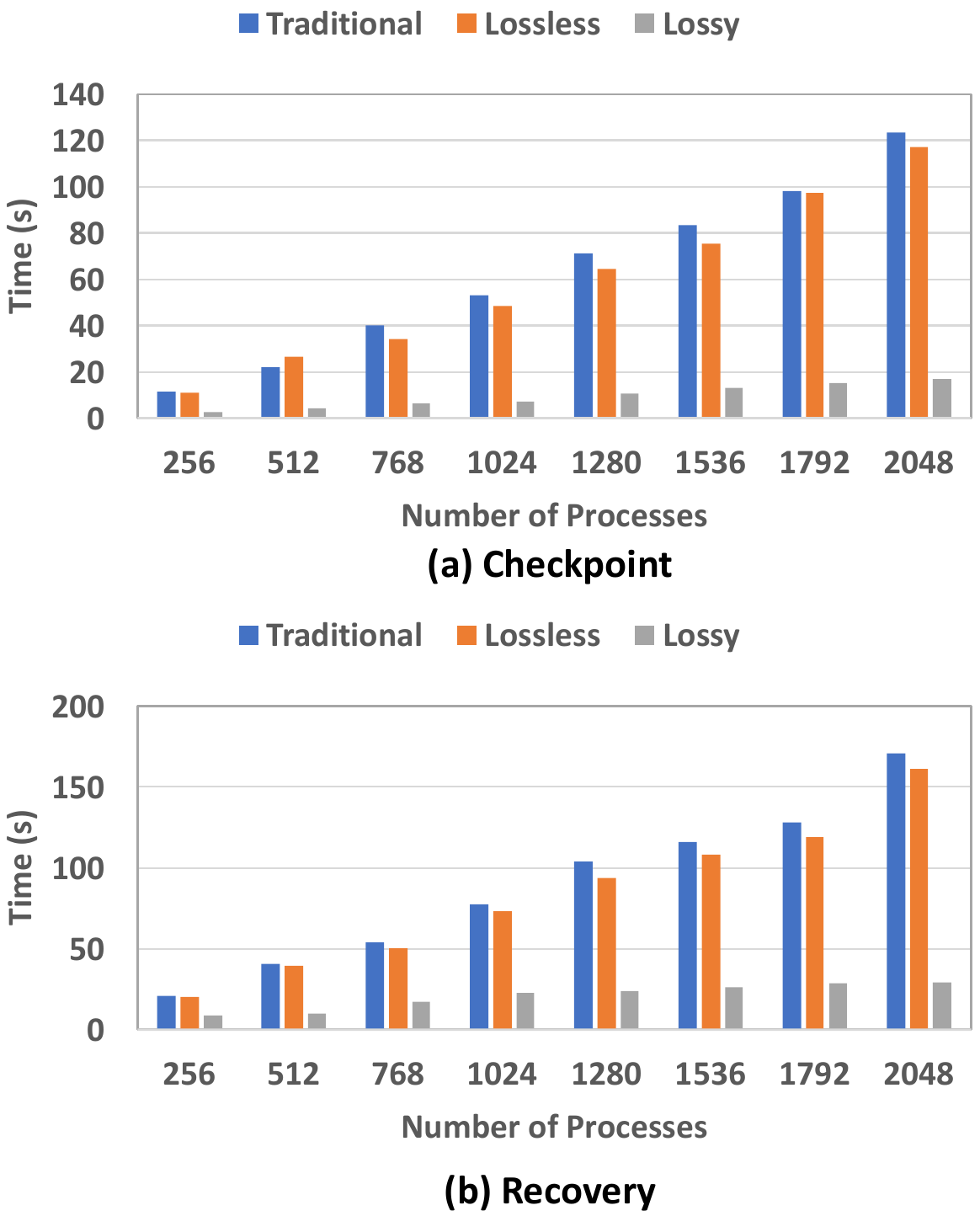}
\caption{Average time of one checkpoint and recovery for CG with different checkpointing techniques on Bebop.}
\label{fig:f9}
\end{figure}

Based on the evaluated checkpointing time for different iterative methods (as shown in Figure \ref{fig:f7}, \ref{fig:f3}, and \ref{fig:f9}) and Equation (\ref{eq:lossy-overhead}), we can theoretically analyze the expected fault tolerance overhead for Jacobi, GMRES, and CG with two failure rates (i.e., MTTI = $1$ hour and MTTI = $3$ hours), as shown in Figure \ref{fig:f2}.
Note that for the Jacobi method, the expectation of $N'$ is about $6$, which is calculated based on the interval
$[\mathit{\frac{N+1}{2}-\log_R(R^{\frac{N+1}{2}}+eb)}$, $\mathit{N-\log_R(R^N+eb)}]$,
where $N = 3941$ and $eb = 10^{-4}$.
We estimate the spectral radius $R$ based on the final relative norm error and the number of convergence iterations. In our experiments, $R \approx 0.99998$.
Following the discussion in Section \ref{sec:impact}, we set $N'$ to $0$ for GMRES and $594$ for CG (i.e., $25\%$ of the CG's total iterations). 

\begin{figure}[]
\centering
\includegraphics[scale=0.57]{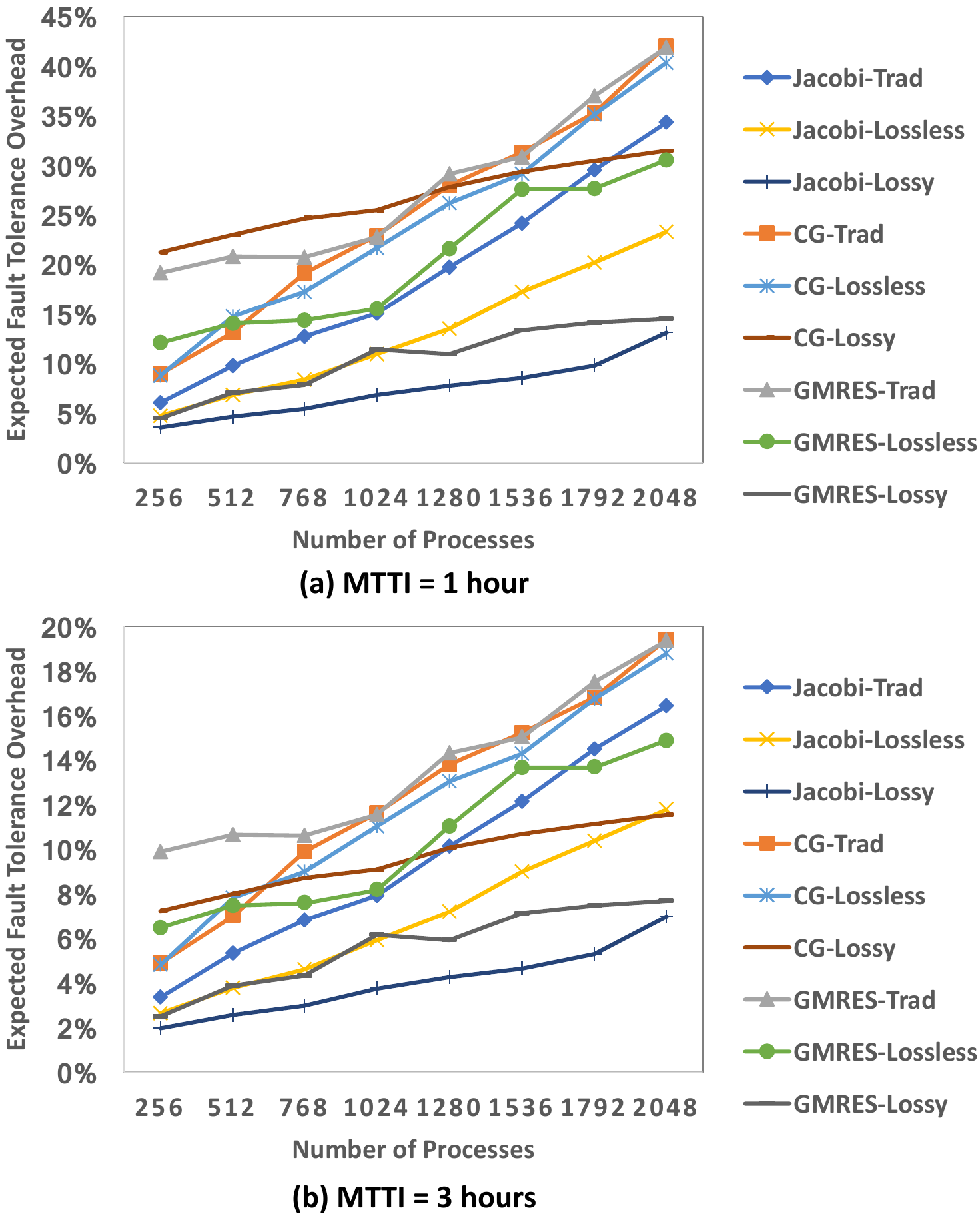}
\caption{Expected overhead of fault-tolerant Jacobi, GMRES, \& CG with different checkpointing techniques and failure rates on Bebop.}
\label{fig:f2}
\end{figure}

Figure \ref{fig:f2} illustrates that for both failure rates, the expected fault tolerance overhead of our proposed lossy checkpointing is always better than that of the other two solutions for Jacob and GMRES.
For CG, the expected overhead of lossy checkpointing is better than that of the other two solutions when the number of processes is greater than $1536$ and $768$ for the two failure rates, respectively. 
We note that in Figure \ref{fig:f2}, the curves with lossy checkpointing increase much slowly than the curves with other two checkpointing solutions, thus demonstrating that our proposed lossy checkponting is expected to achieve more performance gain as scale increases compared with the other two solutions. 
In the next subsection we will use the optimal checkpoint interval with given failure rate to experimentally prove this conclusion.

\subsection{Experimental Evaluation}
In this subsection, we evaluate the fault tolerance overhead experimentally for the three solutions with their corresponding optimal checkpointing intervals in the presence of injected failures.
As described in Section \ref{sec:related}, the MTTI can be almost hourly; hence, we inject failures with the rate being one failure per hour (i.e., $T_f = 3600$ seconds) in the experiment.
Each failure may occur randomly at any time, including during computations of iterative methods and in the checkpoint/recovery period.
The failure intervals follow an exponential distribution, because this is a common behavior of a system for most of its lifetime.
According to Young's formula (as shown in Equation (\ref{eq:young})), 
we can calculate the optimal checkpointing interval for the three solutions based on this failure rate and their checkpointing time as shown in Figure \ref{fig:f3}.
Specifically, the calculated optimal checkpoint intervals for the traditional, lossless, and lossy checkpointing are $16$ minutes, $12$ minutes, and $7$ minutes, respectively.
We run each case with 2,048 processes/cores on Bebop ten times and investigate the average overall running time.
The baseline time of the iterative methods is the overall productive execution time of solving the 3D Poisson equation (as shown in Equation (\ref{eq:linear-system})) \textit{once} without checkpointing and failure interruption. Specifically, the baseline times of Jacobi, GMRES, and CG are about 50 minutes, 120 minutes, and 35 minutes, respectively.
We also compare the experimental overhead with the expected overhead derived theoretically by our performance model.

\begin{figure}[]
\centering
\includegraphics[scale=0.53]{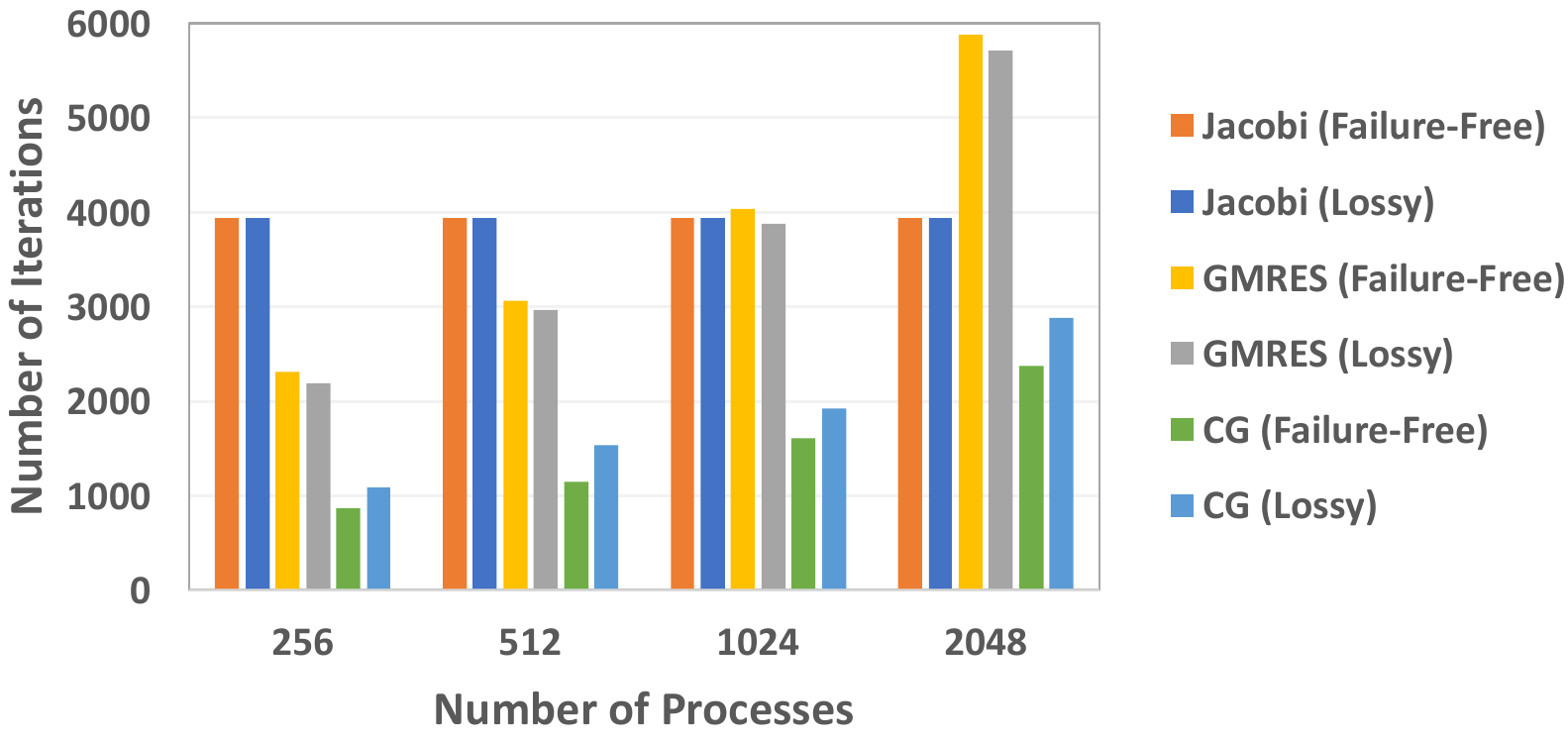}
\caption{Number of convergence iterations with lossy checkpointing method for Jacobi, GMRES, and CG on Bebop.}
\label{fig:f6}
\end{figure}

\begin{figure}[]
\centering
\includegraphics[scale=0.36]{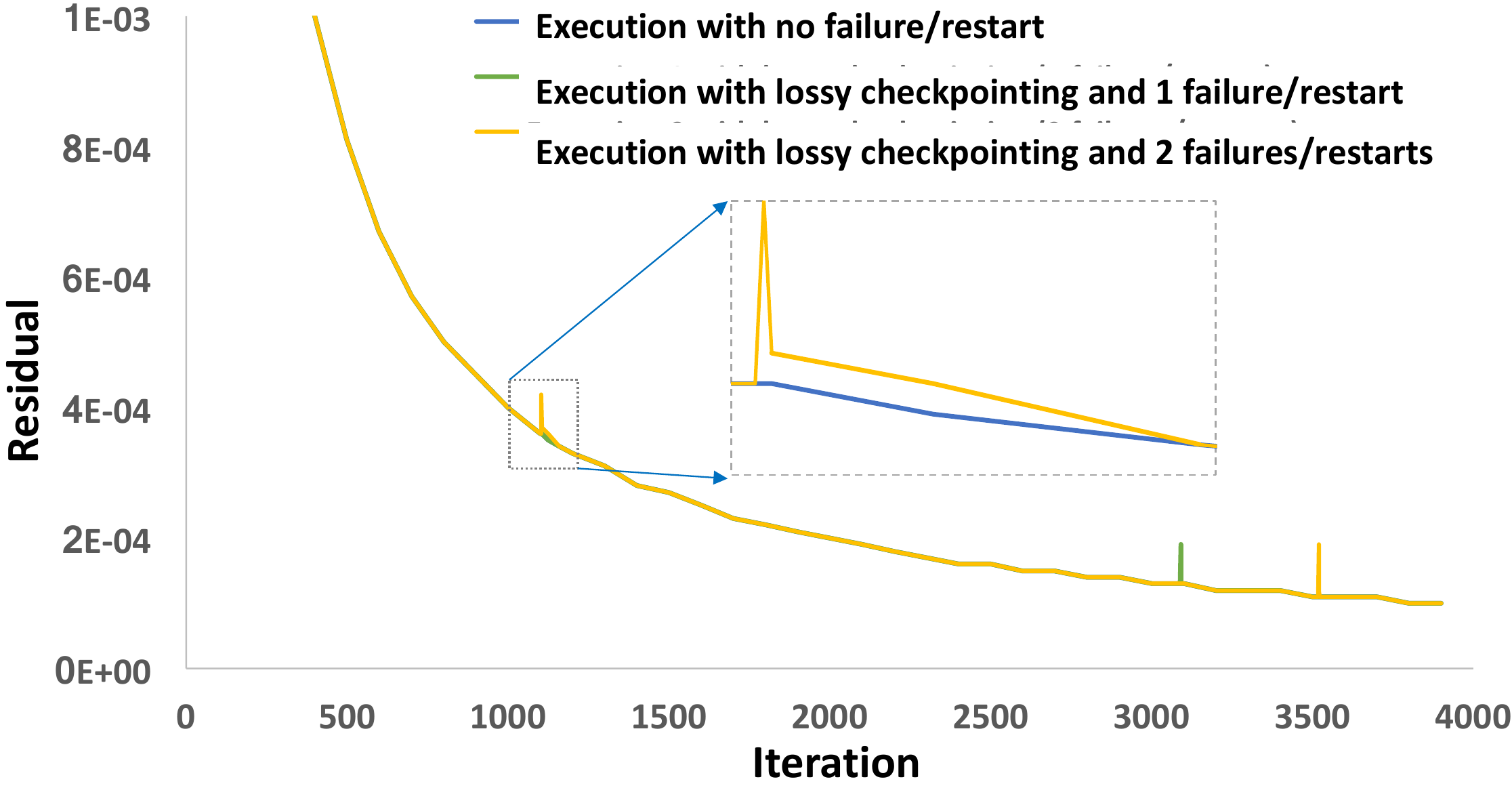}
\caption{Typical example executions of Jacobi method with lossy checkpointing on Bebop.}
\label{fig:f8}
\end{figure}

Figure \ref{fig:f6} presents the numbers of convergence iterations with lossy checkpointing for Jacobi, GMRES, and CG compared with their baseline executions (failure-free) on the Bebop cluster with 2,048 processes.
The experiments illustrate that lossy checkpointing under our settings (including convergence accuracy, error bound and failure rate) introduces no delay (i.e., $0$ extra iterations) on the convergence for Jacobi method, as shown in Figure \ref{fig:f6}. 
This is consistent with our theoretical analysis in Section \ref{sec:stat-bound}: the upper bound of the number of extra iterations $N'$ is $6$, based on the interval $[\mathit{\frac{N+1}{2}-\log_R(R^{\frac{N+1}{2}}+eb)}, \mathit{N-\log_R(R^N+eb)}]$ with $N = 3941$ and $eb = 10^{-4}$.  
Figure \ref{fig:f6} shows that lossy checkpointing slightly accelerates the convergence of GMRES in the condition of bounding the jump of the residual (Theorem 3), which is consistent with the analysis and discussion presented in Section \ref{sec:gmres-bound}.
Figure \ref{fig:f6} also illustrates that the lossy checkpointing with $eb = 10^{-4}$ and $T_f = 3600$ seconds will delay the convergence of CG by $24.8\%$ on average in terms of the convergence iterations, which is consistent with the analysis presented in Section \ref{sec:cg-bound} (as shown in Figure \ref{fig:f5}).

In Figure \ref{fig:f8}, we show two typical example executions of Jacobi method with lossy checkpointing. It shows that after a lossy recovery, Jacobi method can quickly converge to the same residual value as the failure-free Jacobi does, with no extra iterations.

\begin{figure}[]
\centering
\includegraphics[scale=0.52]{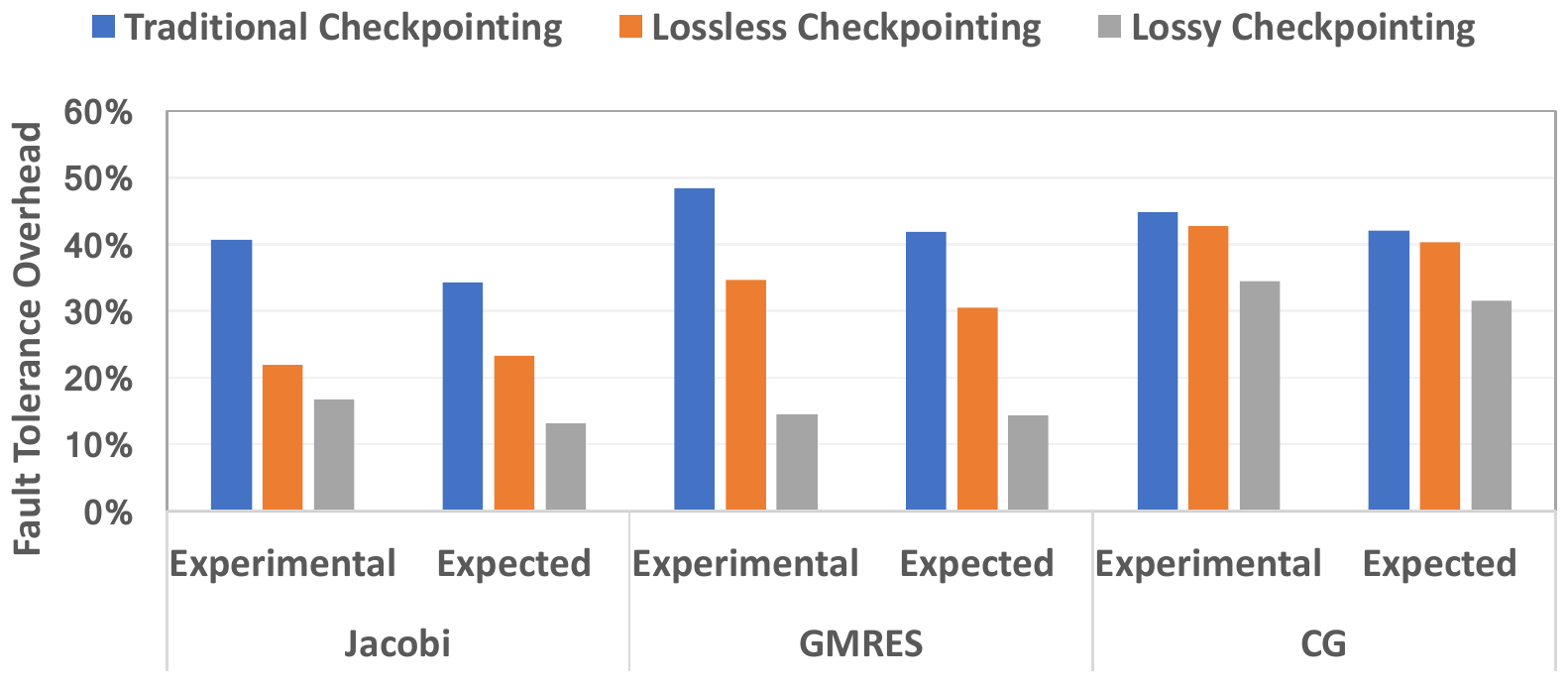}
\caption{Experimental overhead versus expected overhead of fault-tolerant Jacobi, GMRES, and CG with different checkpointing techniques on Bebop.}
\label{fig:f4}
\end{figure}

In Figure \ref{fig:f4} we present the average fault tolerance overhead of these three solutions with 2,048 processes on Bebop.
Here the fault tolerance overhead refers to the overall running time taking away the baseline time. 
The figure illustrates that our proposed lossy checkpointing outperforms the other two solutions with Jacobi, GMRES, and CG.
Specifically, for Jacobi, our solution reduces the fault tolerance overhead by $59\%$ compared with the traditional checkpointing and $24\%$ compared with the lossless checkpointing. 
For GMRES, our solution outperforms the traditional checkpointing and the lossless checkpointing by $70\%$ and $58\%$, respectively, in terms of the fault tolerance overhead.
For CG, our solution reduce the fault tolerance overhead by $23\%$ and $20\%$ compared with the traditional and lossless checkpointing, respectively.
Note that the experimental overheads for traditional and lossless checkpointing are higher than their expected overheads except for Jacobi with lossless checkpointing.
The reason could be that the expected overhead formulas (Equations (\ref{eq:lossless-overhead}) and (\ref{eq:lossy-overhead})) assume that the checkpointing time equals the recovery time, namely, $T_{rc} = T_{ckp}$.
Except for Jacobi with lossless checkpointing, however, the other solutions' recovery time is higher than their checkpointing time, as shown in Figures \ref{fig:f7}, \ref{fig:f3}, and \ref{fig:f9}, because of reconstructing static variables.
Hence, it will lead to a loss of accuracy between the experimental results and the expected analysis in terms of the fault tolerance overhead. 
Note that for our lossy checkpointing, there is only a small difference (up to about $10\%$) between the experimental overhead and the expected overhead, as shown in Figure \ref{fig:f4}.

\section{Conclusion and Future Work}
\label{sec:conclusion}

In this paper, we propose a novel, efficient lossy checkpointing scheme, by exploring how to efficiently leverage the lossy compression technique to improve the overall checkpointing/restart performance for iterative methods in failure prone environment. We have four significant contributions: 
(1) we propose a lossy checkpointing scheme that can significantly improve the checkpointing performance for iterative methods; (2) we formulate the lossy checkpointing performance model and quantify the tradeoff between the reduced checkpointing overhead and the extra number of iterations caused by the compression errors; (3) we analyze the impact of the lossy checkpointing for multiple types of iterative methods; and (4) we evaluate the lossy checkpointing solution using a parallel environment with 2,048 cores. Our experiments show that our lossy checkpointing method can significantly reduce the fault tolerance overhead for the Jacobi, GMRES, and CG methods in the presence of failures, by $20\%{\sim}58\%$ compared with traditional checkpointing and more than $23\%{\sim}70\%$ compared with lossless compressed checkpointing. We plan to study how to extend our lossy checkpointing scheme to additional scientific application domains.
\section*{Acknowledgments}
\small
This research was supported by the Exascale Computing Project (ECP), Project Number: 17-SC-20-SC, a collaborative effort of two DOE organizations -- the Office of Science and the National Nuclear Security Administration, responsible for the planning and preparation of a capable exascale ecosystem, including software, applications, hardware, advanced system engineering and early testbed platforms, to support the nation's exascale computing imperative. The material was supported by the U.S. Department of Energy, Office of Science, under contract DE-AC02-06CH11357, and supported by the National Science Foundation under Grant No. 1305624, No. 1513201, and No. 1619253. We gratefully acknowledge the computing resources provided on Bebop, a high-performance computing cluster operated by the Laboratory Computing Resource Center at Argonne National Laboratory. We would like to thank Dr. Patrick Bridges for his helpful suggestions for the final paper.

\bibliographystyle{abbrv}
\bibliography{bib/refs}

\begin{thebibliography}{10}

\bibitem{bebop}
Bebop cluster.
\newblock \url{https://www.lcrc.anl.gov/systems/resources/bebop}, 2018.
\newblock Online.

\bibitem{petsc-linear-solver}
{Summary of Sparse Linear Solvers Available from PETSc}.
\newblock
  \url{http://www.mcs.anl.gov/petsc/documentation/linearsolvertable.html},
  2018.
\newblock Online.

\bibitem{agbaria1999starfish}
A.~M. Agbaria and R.~Friedman.
\newblock Starfish: Fault-tolerant dynamic mpi programs on clusters of
  workstations.
\newblock In {\em Proceedings of 8th International Symposium on High
  Performance Distributed Computing.}, pages 167--176, 1999.

\bibitem{agullo2013towards}
E.~Agullo, L.~Giraud, A.~Guermouche, J.~Roman, and M.~Zounon.
\newblock {\em Towards resilient parallel linear Krylov solvers:
  recover-restart strategies}.
\newblock PhD thesis, INRIA, 2013.

\bibitem{tucker}
W.~Austin, G.~Ballard, and T.~G. Kolda.
\newblock Parallel tensor compression for large-scale scientific data.
\newblock In {\em 2016 IEEE International Parallel and Distributed Processing
  Symposium}, pages 912--922, 2016.

\bibitem{bahi2007parallel}
J.~M. Bahi, S.~Contassot-Vivier, and R.~Couturier.
\newblock {\em Parallel iterative algorithms: from sequential to grid
  computing}.
\newblock CRC Press, 2007.

\bibitem{petsc}
S.~Balay, S.~Abhyankar, M.~Adams, J.~Brown, P.~Brune, K.~Buschelman, L.~Dalcin,
  V.~Eijkhout, W.~Gropp, D.~Kaushik, et~al.
\newblock Petsc users manual revision 3.8.
\newblock Technical report, Argonne National Lab.(ANL), Argonne, IL (United
  States), 2017.

\bibitem{barrett1994templates}
R.~Barrett, M.~Berry, T.~F. Chan, J.~Demmel, J.~Donato, J.~Dongarra,
  V.~Eijkhout, R.~Pozo, C.~Romine, and H.~Van~der Vorst.
\newblock {\em Templates for the solution of linear systems: building blocks
  for iterative methods}.
\newblock SIAM, 1994.

\bibitem{parkere2013role}
R.~Barrett, S.~Borkar, S.~Dosanjh, S.~Hammond, M.~Heroux, X.~Hu, J.~Luitjens,
  S.~Parker, J.~Shalf, and L.~Tang.
\newblock On the role of co-design in high performance computing.
\newblock {\em Transition of HPC Towards Exascale Computing}, 24:141, 2013.

\bibitem{fti}
L.~Bautista-Gomez, S.~Tsuboi, D.~Komatitsch, F.~Cappello, N.~Maruyama, and
  S.~Matsuoka.
\newblock Fti: High performance fault tolerance interface for hybrid systems.
\newblock In {\em Proceedings of 2011 International Conference for High
  Performance Computing, Networking, Storage and Analysis}, page~32, 2011.

\bibitem{becciani2014solving}
U.~Becciani, E.~Sciacca, M.~Bandieramonte, A.~Vecchiato, B.~Bucciarelli, and
  M.~G. Lattanzi.
\newblock Solving a very large-scale sparse linear system with a parallel
  algorithm in the gaia mission.
\newblock In {\em High Performance Computing \& Simulation (HPCS), 2014
  International Conference on}, pages 104--111. IEEE, 2014.

\bibitem{bluewater}
B.~Bode, M.~Butler, T.~Dunning, W.~Gropp, T.~Hoefler, W.-m. Hwu, and W.~Kramer.
\newblock The {Blue Waters} super-system for super-science.
\newblock {\em Contemporary High Performance Computing Architectures}, 2012.

\bibitem{bridges2012fault}
P.~G. Bridges, K.~B. Ferreira, M.~A. Heroux, and M.~Hoemmen.
\newblock Fault-tolerant linear solvers via selective reliability.
\newblock pages 914--922, 2015.

\bibitem{calhoun}
J.~Calhoun, F.~Cappello, L.~Olson, M.~Snir, and W.~Gropp.
\newblock Exploring the feasibility of lossy compression for pde simulations.
\newblock {\em The International Journal of High Performance Computing
  Applications}, 2018.
\newblock To appear.

\bibitem{chen2016online}
J.~Chen, X.~Liang, and Z.~Chen.
\newblock Online algorithm-based fault tolerance for cholesky decomposition on
  heterogeneous systems with gpus.
\newblock In {\em 2016 IEEE International Parallel and Distributed Processing
  Symposium}, pages 993--1002. IEEE, 2016.

\bibitem{chen2011algorithm}
Z.~Chen.
\newblock Algorithm-based recovery for iterative methods without checkpointing.
\newblock In {\em Proceedings of the 20th International Symposium on High
  Performance Distributed Computing}, pages 73--84, 2011.

\bibitem{chen2013online}
Z.~Chen.
\newblock Online-abft: An online algorithm based fault tolerance scheme for
  soft error detection in iterative methods.
\newblock In {\em Proceedings of the 18th ACM SIGPLAN Symposium on Principles
  and Practice of Parallel Programming}, volume~48, pages 167--176, 2013.

\bibitem{numarck}
Z.~Chen, S.~W. Son, W.~Hendrix, A.~Agrawal, W.-k. Liao, and A.~Choudhary.
\newblock Numarck: Machine learning algorithm for resiliency and checkpointing.
\newblock In {\em Proceedings of the International Conference for High
  Performance Computing, Networking, Storage and Analysis}, pages 733--744,
  2014.

\bibitem{chorin1968numerical}
A.~J. Chorin.
\newblock Numerical solution of the navier-stokes equations.
\newblock {\em Mathematics of computation}, 22(104):745--762, 1968.

\bibitem{ufl}
T.~A. Davis and Y.~Hu.
\newblock The university of florida sparse matrix collection.
\newblock {\em ACM Transactions on Mathematical Software}, 38(1):1, 2011.

\bibitem{gzip}
L.~P. Deutsch.
\newblock Gzip file format specification version 4.3.
\newblock 1996.

\bibitem{di2014optimization-2}
S.~Di, L.~Bautista-Gomez, and F.~Cappello.
\newblock Optimization of a multilevel checkpoint model with uncertain
  execution scales.
\newblock In {\em Proceedings of the International Conference for High
  Performance Computing, Networking, Storage and Analysis}, pages 907--918,
  2014.

\bibitem{di2014optimization-1}
S.~Di, M.~S. Bouguerra, L.~Bautista-Gomez, and F.~Cappello.
\newblock Optimization of multi-level checkpoint model for large scale hpc
  applications.
\newblock In {\em 28th International Parallel and Distributed Processing
  Symposium}, pages 1181--1190, 2014.

\bibitem{di2016adaptive}
S.~Di and F.~Cappello.
\newblock Adaptive impact-driven detection of silent data corruption for hpc
  applications.
\newblock {\em IEEE Transactions on Parallel and Distributed Systems},
  27(10):2809--2823, 2016.

\bibitem{sz16}
S.~Di and F.~Cappello.
\newblock Fast error-bounded lossy hpc data compression with sz.
\newblock In {\em 2016 IEEE International Parallel and Distributed Processing
  Symposium}, pages 730--739. IEEE, 2016.

\bibitem{ftgmres}
J.~Elliott, M.~Hoemmen, and F.~Mueller.
\newblock Evaluating the impact of sdc on the gmres iterative solver.
\newblock In {\em 2014 IEEE 28th International Parallel and Distributed
  Processing Symposium}, pages 1193--1202, 2014.

\bibitem{fu2016sunway}
H.~Fu, J.~Liao, J.~Yang, L.~Wang, Z.~Song, X.~Huang, C.~Yang, W.~Xue, F.~Liu,
  F.~Qiao, et~al.
\newblock The sunway taihulight supercomputer: system and applications.
\newblock {\em Science China Information Sciences}, 59(7):072001, 2016.

\bibitem{heath2002scientific}
M.~T. Heath.
\newblock {\em Scientific computing}.
\newblock McGraw-Hill New York, 2002.

\bibitem{islam}
T.~Z. Islam, K.~Mohror, S.~Bagchi, A.~Moody, B.~R. De~Supinski, and
  R.~Eigenmann.
\newblock {McrEngine}: A scalable checkpointing system using data-aware
  aggregation and compression.
\newblock In {\em Proceedings of the International Conference on High
  Performance Computing, Networking, Storage and Analysis}, page~17, 2012.

\bibitem{isabela}
S.~Lakshminarasimhan, N.~Shah, S.~Ethier, S.-H. Ku, C.-S. Chang, S.~Klasky,
  R.~Latham, R.~Ross, and N.~F. Samatova.
\newblock Isabela for effective in situ compression of scientific data.
\newblock {\em Concurrency and Computation: Practice and Experience},
  25(4):524--540, 2013.

\bibitem{langou}
J.~Langou, Z.~Chen, G.~Bosilca, and J.~Dongarra.
\newblock Recovery patterns for iterative methods in a parallel unstable
  environment.
\newblock {\em SIAM Journal on Scientific Computing}, 30(1):102--116, 2007.

\bibitem{li2015applicationspecific}
G.~Li, K.~Pattabiraman, C.-Y. Cher, and P.~Bose.
\newblock An applicationspecific checkpointing technique for minimizing
  checkpoint corruption.
\newblock In {\em International Symposium on Software Reliability Engineering}.
  IEEE, 2015.

\bibitem{liang2017correcting}
X.~Liang, J.~Chen, D.~Tao, S.~Li, P.~Wu, H.~Li, K.~Ouyang, Y.~Liu, F.~Song, and
  Z.~Chen.
\newblock Correcting soft errors online in fast fourier transform.
\newblock In {\em Proceedings of the International Conference for High
  Performance Computing, Networking, Storage and Analysis}, page~30. ACM, 2017.

\bibitem{zfp}
P.~Lindstrom.
\newblock Fixed-rate compressed floating-point arrays.
\newblock {\em IEEE Transactions on Visualization and Computer Graphics},
  20(12):2674--2683, 2014.

\bibitem{fpzip}
P.~Lindstrom and M.~Isenburg.
\newblock Fast and efficient compression of floating-point data.
\newblock {\em IEEE Transactions on Visualization and Computer Graphics},
  12(5):1245--1250, 2006.

\bibitem{liu2018large}
R.-T. Liu and Z.-N. Chen.
\newblock A large-scale study of failures on petascale supercomputers.
\newblock {\em Journal of Computer Science and Technology}, 33(1):24--41, 2018.

\bibitem{scr}
A.~Moody, G.~Bronevetsky, K.~Mohror, and B.~R.~d. Supinski.
\newblock Design, modeling, and evaluation of a scalable multi-level
  checkpointing system.
\newblock In {\em Proceedings of the 2010 ACM/IEEE International Conference for
  High Performance Computing, Networking, Storage and Analysis}, pages 1--11,
  2010.

\bibitem{mora2001numerical}
J.~Mora~Acosta.
\newblock {\em Numerical algorithms for three dimensional computational fluid
  dynamic problems}.
\newblock Universitat Polit{\`e}cnica de Catalunya, 2001.

\bibitem{patankar1980numerical}
S.~Patankar.
\newblock {\em Numerical heat transfer and fluid flow}.
\newblock CRC press, 1980.

\bibitem{diskless}
J.~S. Plank, K.~Li, and M.~A. Puening.
\newblock Diskless checkpointing.
\newblock {\em IEEE Transactions on Parallel and Distributed Systems},
  9(10):972--986, 1998.

\bibitem{powell1977restart}
M.~J.~D. Powell.
\newblock Restart procedures for the conjugate gradient method.
\newblock {\em Mathematical programming}, 12(1):241--254, 1977.

\bibitem{lossless2006}
P.~Ratanaworabhan, J.~Ke, and M.~Burtscher.
\newblock Fast lossless compression of scientific floating-point data.
\newblock In {\em 2006 Data Compression Conference.}, pages 133--142, 2006.

\bibitem{gmres}
Y.~Saad and M.~H. Schultz.
\newblock Gmres: A generalized minimal residual algorithm for solving
  nonsymmetric linear systems.
\newblock {\em SIAM Journal on Scientific and Statistical Computing},
  7(3):856--869, 1986.

\bibitem{sao2013self}
P.~Sao and R.~Vuduc.
\newblock Self-stabilizing iterative solvers.
\newblock In {\em Proceedings of the Workshop on Latest Advances in Scalable
  Algorithms for Large-Scale Systems}, page~4, 2013.

\bibitem{ssem}
N.~Sasaki, K.~Sato, T.~Endo, and S.~Matsuoka.
\newblock Exploration of lossy compression for application-level
  checkpoint/restart.
\newblock In {\em 2015 IEEE International Parallel and Distributed Processing
  Symposium}, pages 914--922, 2015.

\bibitem{nlpkkt240}
O.~Schenk.
\newblock {Symmetric indefinite KKT matrices}.
\newblock \url{https://sparse.tamu.edu/Schenk}, 2018.
\newblock Online.

\bibitem{nlpkkt}
O.~Schenk, A.~W{\"a}chter, and M.~Weiser.
\newblock Inertia-revealing preconditioning for large-scale nonconvex
  constrained optimization.
\newblock {\em SIAM Journal on Scientific Computing}, 31(2):939--960, 2008.

\bibitem{sz17}
D.~Tao, S.~Di, Z.~Chen, and F.~Cappello.
\newblock Significantly improving lossy compression for scientific data sets
  based on multidimensional prediction and error-controlled quantization.
\newblock In {\em 2017 IEEE International Parallel and Distributed Processing
  Symposium}, pages 1129--1139. IEEE, 2017.

\bibitem{fti-sz}
D.~Tao, S.~Di, X.~Liang, Z.~Chen, and F.~Cappello.
\newblock {Lossy Checkpointing Library}.
\newblock \url{https://github.com/dingwentao/fti-sz}, 2018.
\newblock Online.

\bibitem{newsum}
D.~Tao, S.~L. Song, S.~Krishnamoorthy, P.~Wu, X.~Liang, E.~Z. Zhang,
  D.~Kerbyson, and Z.~Chen.
\newblock New-sum: A novel online abft scheme for general iterative methods.
\newblock In {\em Proceedings of the 25th ACM International Symposium on
  High-Performance Parallel and Distributed Computing}, pages 43--55, 2016.

\bibitem{mpiio}
R.~Thakur, W.~Gropp, and E.~Lusk.
\newblock On implementing mpi-io portably and with high performance.
\newblock In {\em Proceedings of the Sixth Workshop on I/O in Parallel and
  Distributed Systems}, pages 23--32. ACM, 1999.

\bibitem{wu2014ft}
P.~Wu and Z.~Chen.
\newblock Ft-scalapack: Correcting soft errors on-line for scalapack cholesky,
  qr, and lu factorization routines.
\newblock In {\em Proceedings of the 23rd International Symposium on
  High-Performance Parallel and Distributed Computing}, pages 49--60. ACM,
  2014.

\bibitem{wu2017silent}
P.~Wu, N.~DeBardeleben, Q.~Guan, S.~Blanchard, J.~Chen, D.~Tao, X.~Liang,
  K.~Ouyang, and Z.~Chen.
\newblock Silent data corruption resilient two-sided matrix factorizations.
\newblock In {\em Proceedings of the 22nd ACM SIGPLAN Symposium on Principles
  and Practice of Parallel Programming}, pages 415--427. ACM, 2017.

\bibitem{wu2016towards}
P.~Wu, Q.~Guan, N.~DeBardeleben, S.~Blanchard, D.~Tao, X.~Liang, J.~Chen, and
  Z.~Chen.
\newblock Towards practical algorithm based fault tolerance in dense linear
  algebra.
\newblock In {\em Proceedings of the 25th ACM International Symposium on
  High-Performance Parallel and Distributed Computing}, pages 31--42. ACM,
  2016.

\bibitem{young}
J.~W. Young.
\newblock A first order approximation to the optimum checkpoint interval.
\newblock {\em Communications of the ACM}, 17(9):530--531, 1974.

\end{thebibliography}

\end{document}